\documentclass[a4paper,USenglish,cleveref, autoref, thm-restate, numberwithinsect]{lipics-v2021}
\pdfoutput=1
\hideLIPIcs  %
\nolinenumbers %uncomment to disable line numbering

\title{Functional Closure Properties of Finite $\IN$-weighted Automata}

\author{Julian Dörfler}{Saarland Informatics Campus (SIC), Saarbrücken Graduate School of Computer Science, Saarland University, Germany}{jdoerfler@cs.uni-saarland.de}{https://orcid.org/0000-0002-0943-8282}{}
\author{Christian Ikenmeyer}{University of Warwick, United Kingdom}{christian.ikenmeyer@warwick.ac.uk}{https://orcid.org/0000-0003-4654-177X}{EPSRC EP/W014882/2}

\authorrunning{J. Dörfler and C. Ikenmeyer} %

\Copyright{Julian Dörfler and Christian Ikenmeyer} %



\category{} %

\relatedversion{} %

\acknowledgements{}%

\EventEditors{}%
\EventNoEds{2}%
\EventLongTitle{}%
\EventShortTitle{}%
\EventAcronym{}%
\EventYear{}%
\EventDate{}%
\EventLocation{}%
\EventLogo{}
\SeriesVolume{}%
\ArticleNo{}%

\usepackage{graphicx,latexsym,amsmath,amsthm,amssymb,color,verbatim,mathrsfs,bbm,amscd,placeins,hyperref,caption,subcaption}
\usepackage[utf8] {inputenc}
\usepackage{tikz}
\usetikzlibrary{shapes,arrows,fit,calc,positioning,automata,decorations.markings}
\usepackage{esvect}

\newcommand{\automatastyle}[2]{\tikzstyle{every node}=[auto, node distance=#1mm]
   \tikzstyle{inner}=[thick, minimum width=#2mm, circle, draw=black]
   \tikzstyle{start}=[thick, minimum width=#2mm, circle, draw=black, pin=left:]
   \tikzstyle{accept}=[thick, minimum width=#2mm, circle, draw=black, double]
   \tikzstyle{every pin edge}=[stealth-]
   \tikzstyle{every pin}=[pin distance=5mm]
   \tikzstyle{every loop}=[-stealth]}

\automatastyle{25}{8}

\newcommand{\IN}{\mathbb{N}}
\newcommand{\IQ}{\mathbb{Q}}
\newcommand{\IZ}{\mathbb{Z}}
\newcommand{\IR}{\mathbb{R}}

\newcommand{\CR}{\mathcal{R}}

\newcommand{\eps}{\varepsilon}
\newcommand{\wt}{\textup{wt}}
\newcommand{\inp}{\textup{in}}
\newcommand{\outp}{\textup{out}}
\newcommand{\powset}{\mathcal{P}}
\newcommand{\Cweight}{{\ensuremath{\mathbf{w}}}}
\newcommand{\Cpweight}{{\ensuremath{\underline{\mathbf{w}}}}}

\newcommand{\ONE}{\mathbbm{1}}

\newcommand{\prop}{\textup{prop}}
\newcommand{\init}{\textup{init}}
\newcommand{\step}{\textup{step}}
\newcommand{\cond}{\textup{cond}}
\newcommand{\bin}{\textup{bin}}
\newcommand{\smod}[3]{\ensuremath{#1 \operatorname{srem}_{#3} #2}} %
\newcommand{\smodtimesm}[3]{\ensuremath{#1 \operatorname{srem}^{\times m}_{#3} #2}} %

\newcommand{\NP}{\ensuremath{\mathsf{NP}}}
\newcommand{\sharpFA}{\ensuremath{\mathsf{\#FA}}}
\newcommand{\sharpP}{\ensuremath{\mathsf{\#P}}}

\DeclareMathOperator{\image}{im}
\DeclareMathOperator{\rem}{rem}

\makeatletter
\newsavebox{\@brx}
\newcommand{\llangle}[1][]{\savebox{\@brx}{\(\m@th{#1\langle}\)}%
  \mathopen{\copy\@brx\kern-0.5\wd\@brx\usebox{\@brx}}}
\newcommand{\rrangle}[1][]{\savebox{\@brx}{\(\m@th{#1\rangle}\)}%
  \mathclose{\copy\@brx\kern-0.5\wd\@brx\usebox{\@brx}}}
\makeatother

\allowdisplaybreaks[1]

\newenvironment{proofsketch}{%
  \renewcommand{\proofname}{Proof sketch}\proof}{\endproof}

\begin{document}
\sloppy

\maketitle

\begin{abstract}
We determine all functional closure properties of finite $\IN$-weighted automata, even all multivariate ones, and in particular all multivariate polynomials.
We also determine all univariate closure properties in the promise setting, and all multivariate closure properties under certain assumptions on the promise, in particular we determine all multivariate closure properties where the output vector lies on a monotone algebraic graph variety.
\end{abstract}

\newpage
\setcounter{page}{1}

\section{Finite $\IN$-weighted automata and functional closure properties}
\label{sec:intro}
Let $\Sigma$ be a finite set, for example $\Sigma=\{0,1\}$.
A finite $\IN$-weighted automaton with all weights 1 is a nondeterministic finite automaton that on input $w\in\Sigma^\star$ outputs the number of accepting computation paths on input $w$, instead of just outputting whether or not an accepting computation path exists, see Def.~\ref{def:NFA} for the formal definition \footnote{We use the equality of the number of accepting paths of an NFA and the output of the corresponding $\IN$-weighted automaton, see \cite[Exa.~2.2]{DK21}.}.
While every nondeterministic finite automaton determines a subset of $\Sigma^\star$,
a finite $\IN$-weighted automaton computes a function $\Sigma^\star\to\IN$.
A function $f:\Sigma^\star \to \IN$
can be presented as the series $\sum_{w \in \Sigma^\star} f(w) w$, 
and the set of series is denoted by $\IN\llangle\Sigma^\star\rrangle$ in the automata literature, see e.g.~\cite{DK21}\footnote{A series with finite support is called a \emph{polynomial}, but we will not be concerned with the support of series in this paper. Instead, we use the term \emph{polynomial} as it is used in commutative algebra, and we mean multivariate polynomials with rational coefficients.}.
The natural way of adding two functions $\Sigma^\star\to\IN$
and adding two series in $\IN\llangle\Sigma^\star\rrangle$ coincides,
but in both presentations we have a natural way of taking the product, and those do not coincide:
\begin{enumerate}
\item Pointwise product of functions $\Sigma^\star \to \IN$. This is called the Hadamard product.
\item Convolution of series, called the Cauchy product.
\end{enumerate}
A series $f$ is called \emph{recognizable} if there is a finite $\IN$-weighted automaton that computes~$f$.
The set of recognizable series is denoted by $\IN^{\textup{rec}}\llangle\Sigma^\star\rrangle$ in \cite{DK21}, but we denote it by $\sharpFA$, to emphasise that we undertake a study similar to $\sharpP$ in \cite[Thm~3.13]{HVW95}, \cite[Thm~6]{Bei97}, and recently \cite{IP22}, but instead of polynomial-time Turing machines we study finite automata.
The Kleene-Schützenberger theorem states that $\sharpFA$ is the smallest set that contains all support 1 series and is closed under sums, Cauchy products, and Kleene-iterations (whenever well-defined, a Kleene-iteration is the sum of all Cauchy powers), see \cite[\S4]{DK21}, but we will not need this insight.

In this paper we study the functional closure properties\footnote{See \cite[Sec.~1]{HVW95} for the naming \emph{functional closure property}. A different reasonable name would be \emph{pointwise closure property}.} of $\sharpFA$.
A function $\varphi:\IN^m\to\IN$ is called a \emph{functional closure property} of $\sharpFA$ if
for all $f_1\in\sharpFA$, $f_2\in\sharpFA$, $\ldots$, $f_m\in\sharpFA$ we have that $\varphi(f_1,\ldots,f_m)\in\sharpFA$.
By $\varphi(f_1,\ldots,f_m)$ we mean the function that on input $w\in\Sigma^\star$ outputs $\varphi(f_1(w),\ldots,f_m(w))$.

Classically, one of the simplest functional closure properties of $\sharpFA$ is
$\varphi:\IN^2\to\IN$,
$\varphi(f_1,f_2) = f_1+f_2$.
This is a functional closure property of $\sharpFA$, because given $f_1\in\sharpFA$ and $f_2\in\sharpFA$, we can show $\varphi(f_1,f_2)=f_1+f_2\in\sharpFA$ by an easy construction: The new NFA consists of a copy of the NFA for $f_1$ and a copy of the NFA for $f_2$, and makes an initial nondeterministic choice as to which NFA to run, see Lemma~\ref{lem:add} for the details.

Another classical simple functional closure property of $\sharpFA$ is
$\varphi:\IN^2\to\IN$,
$\varphi(f_1,f_2) = f_1\cdot f_2$.
This corresponds to the Hadamard product.
This is a functional closure property of $\sharpFA$, because given $f_1\in\sharpFA$ and $f_2\in\sharpFA$, we can show $\varphi(f_1,f_2)=f_1\cdot f_2\in\sharpFA$ by the following construction: The new NFA consists of the product NFA of the NFAs for $f_1$ and $f_2$, and the accepting states correspond to pairs of accepting states, see Lemma~\ref{lem:mult} for the details.
This product construction corresponds to the Hadamard product.

The Cauchy product is also a product on the set $\sharpFA$, but
we explain now that the Cauchy product is not ``functional'', and hence it is out of scope for this type of studies.
If $\varphi:\IN^m\to\IN$ is a functional closure property of $\sharpFA$, then 
we can study the corresponding map $\widetilde\varphi : \underbrace{\sharpFA\times\sharpFA\times\cdots\times\sharpFA}_{m\textup{ times}}\to\sharpFA$.
Observe that if $\varphi$ is a functional closure property of $\sharpFA$, then by definition we have that
for all pairs $(w,w')\in \Sigma^\star\times\Sigma^\star$:\\
\mbox{~}if $(f_1(w),\ldots,f_m(w)) = (f_1(w'),\ldots,f_m(w'))$, \ then \ $\widetilde\varphi(f_1,\ldots,f_m)(w)=\widetilde\varphi(f_1,\ldots,f_m)(w')$.\\
Let $\zeta:\sharpFA\to\sharpFA$ denote the Cauchy square.
We use the observation above to show that $\zeta$ is not equal to $\widetilde\varphi$ for any $\varphi:\IN\to\IN$.
Let $m=1$ and $f(w)=1$ if $w=1$, $f(w)=0$ otherwise. Clearly, $f \in \sharpFA$.
Then $\zeta(f)(11)=1$, and $\zeta(f)(w)=0$ for all $w\neq 11$.
In particular $\zeta(f)(0)\neq\zeta(f)(11)$, even though $f(0)=f(11)$. Hence, $\zeta\neq \widetilde\varphi$ for all $\varphi:\IN\to\IN$.

Numerous functional closure properties of $\sharpFA$ exist, for example the safe decrementation $\max\{0,f_1-1\}$, and the binomial coefficient $\binom{f_1}{2}$.
But not all non-negative functions are functional closure properties of $\sharpFA$, for example $(f_1-f_2)^2$ is not, which can be shown using the Pumping Lemma.
In this paper, we determine \emph{all} functional closure properties of $\sharpFA$, see \S\ref{subsec:ourresults} for the detailed statement.

\subsection{Motivation}
Functional closure properties can be studied for many different counting machine models (also for example with different types of oracle access) and different types of input sets.
The first study of this type was done for nondeterministic polynomial-time Turing machines, i.e., the class $\sharpP$, see \cite{HVW95}, \cite{Bei97}, and the recent \cite{IP22}.
Recall that the class $\sharpP$ is the class of functions $f:\Sigma^*\to\IN$ for which a nondeterministic polynomial time Turing machine $M$ exists such that for all $w\in\Sigma^*$ the number of accepting paths for the computation $M(w)$ is exactly $f(w)$.
The papers mentioned above prove that the relativizing multivariate polynomial closure properties are exactly those polynomials that have nonnegative integers in their expansion over the binomial basis, see \cite{IP22}.
A functional closure property $\varphi:\IN^m\to\IN$ of $\sharpP$ is \emph{relativizing} if $\varphi$ is a closure property for all $\sharpP^A$, where $A \subset \Sigma^*$ is some oracle.
The hope is that for simpler models of computation no oracle access is required to determine the functional closure properties, and we show that this is true for $\sharpFA$, see \S\ref{subsec:ourresults}.

Functional closure properties can be used directly to construct combinatorial proofs of equalities and inequalities. For example, Fermat's little theorem states that $p$ divides $a^p-a$.
The quantity $\frac{1}{p}(a^p-a)$ has a combinatorial interpretation, which can be deduced from the fact that $\frac{1}{p}((f_1)^p-f_1)$ is a univariate functional closure property of $\sharpP$, see \cite[Prop.\,7.3.1]{IP22}, which coincides with the original proof \cite{Pet72}, see also \cite[eq.\,(5)]{Pak22}).
On the other hand, if a function is not a functional closure property, then this means in a very strong sense that there is no combinatorial interpretation for the quantity it describes. For example, the Hadamard inequality
(\cite[\S2.13]{HLP52}, \cite[\S2.11]{BB61},
\cite[eq.\,(2)]{IP22})
states that

\vspace{-0.6cm}

\[\textstyle
\det
\left(\begin{smallmatrix}
a_{11} & \cdots & a_{1d}
\\
\vdots & \ddots & \vdots
\\
a_{d1} & \cdots & a_{dd}
\end{smallmatrix}\right)^2 \leq \prod_{i=1}^d (a_{i1}^2 + \cdots + a_{id}^2).
\]

\vspace{-0.2cm}

One could try to prove this by finding a combinatorial interpretation of the difference $\mathcal H\geq 0$ of the right-hand side and the left-hand side, but even for $d=3$ we have that

\vspace{-0.6cm}

\[\textstyle
\varphi(f_1,\ldots,f_9) \ = \ 
(f_1^2 + f_2^2 + f_3^3)
\cdot
(f_4^2 + f_5^2 + f_6^3)
\cdot
(f_7^2 + f_8^2 + f_9^3)
-
\det
\left(\begin{smallmatrix}
f_1 & f_2 & f_3
\\
f_4 & f_5 & f_6
\\
f_7 & f_8 & f_9
\end{smallmatrix}\right)
\]

\vspace{-0.2cm}

is not a 9-variate relativizing functional closure property of $\sharpP$, see \cite[\S7.2]{IP22}.
In particular, if the function is not a closure property of $\sharpP$, then there are instantiations $f_1,\ldots,f_9\in\sharpP$ such that $\varphi(f_1,\ldots,f_9)$ is not in $\sharpP$,
whereas a combinatorial interpretation of $\mathcal H$ should yield $\varphi(f_1,\ldots,f_9)\in\sharpP$.
This does not rule out a more indirect combinatorial proof for the inequality: For example, for proving combinatorially that $(a-1)^2\geq 0$ one could try to interpret the quantity $(a-1)^2$ combinatorially, but
$(f_1-1)^2$ is not a relativizing closure property of $\sharpP$. However,
$
f_1 \cdot (f_1-1)^2 = 6\binom{f_1}{3} + 2 \binom{f_1}{2}
$
is a
relativizing closure property of $\sharpP$ (see \cite[\S2.4]{IP22}).
There is an obvious combinatorial interpretation of $6\binom{a}{3} + 2 \binom{a}{2}$
as counting size 2 and 3 subsets with multiplicity 6 and 2, respectively.
Hence this gives an indirect combinatorial proof for the inequality $(a-1)^2\geq 0$ by providing a combinatorial interpretation for $a(a-1)^2$.

Some inequalities are only true if the inputs satisfy certain constraints.
For example, the \emph{Ahlswede Daykin inequality}, see \cite{AD78}, \cite{AS16}, \cite[\S 1.2(3)]{IP22}:
If
$a_0 b_0 \geq c_0 d_0$
and
$a_0 b_1 \geq c_0 d_1$
and
$a_1 b_0 \geq c_0 d_1$
and
$a_1 b_1 \geq c_1 d_1$,
then
$(c_0+c_1)(d_0+d_1) \geq (a_0+a_1)(b_0+b_1)$.
If all quantities are in $\sharpP$, including the differences $c_0 d_0 - a_0 b_0$, can we conclude that $(c_0+c_1)(d_0+d_1) - (a_0+a_1)(b_0+b_1)$ is in $\sharpP$? This is an example of a promise problem: We are given twelve $\sharpP$ functions 
$a_0,a_1,b_0,b_1,c_0,c_1,d_0,d_1,h_1,h_2,h_3,h_4$
with the guarantee that
$a_0 b_0 + h_1 = c_0 d_0$,
$a_0 b_1 + h_2 = c_0 d_1$,
$a_1 b_0 + h_3 = c_0 d_1$,
$a_1 b_1 + h_4 = c_1 d_1$.
In other words, the 12-dimensional output vector that we get for every $w\in\Sigma^*$ lies on a codimension 4 algebraic subvariety in $\IQ^{12}$.
Recall that an algebraic subvariety is defined as the simultaneous zero set of a set of polynomials.
Since the 4 variables $h_1,\ldots,h_4$ are determined by the other 8, this variety is a so-called \emph{graph} or \emph{graph variety}.
Numerous questions about combinatorial proofs for inequalities from different areas of mathematics can be phrased in the language of graph varieties, see \cite{IP22}.
The idea is to collect the equations for a set $S$ (the variety) into what is called the \emph{vanishing ideal} $I$, i.e., $I=I(S)=\{\varphi \in \IQ[f_1,\ldots,f_m] \mid \forall (f_1,\ldots,f_m) \in S : \varphi(f_1,\ldots,f_m) = 0\}$;
and define the \emph{coordinate ring} $\IQ[S]$ as the quotient ring $\IQ[f_1,\ldots,f_m]/I(S)$, see \cite{CLO13}.
An element in the quotient ring is a coset with respect to the vanishing ideal.
If there exists a representative $\varphi'$ in a coset $\varphi+I$ that is a functional closure property of $\sharpP$, then every function in $\varphi+I$
is a promise closure property of $\sharpP$ on the variety $S$.
It is desirable to also have the opposite direction, but this only holds under some reasonable restrictions on $S$, in particular it holds for all graph varieties.
This is used in \cite[Prop.\ 2.5.1]{IP22} to show that $c_0 d_0 - a_0 b_0$ is not a relativizing promise closure property of $\sharpP$ on this graph variety.
We prove the same strong dichotomy for monotone graph varieties for $\sharpFA$ instead of $\sharpP$, see Theorem~\ref{thm:mongraphclosures}.

The systematic study of combinatorial interpretations and combinatorial proofs via definitions from computational complexity theory is a very recent research direction \cite{Pak18, Pak19, IP22, Pak22, IPP23, CP23, CP23b}.
The goal is to determine whether or not certain quantities admit a combinatorial description or not. Famous open questions of this type in algebraic combinatorics have been listed by Stanley in \cite{Sta99}, for example his problems 9, 10, and 12.
As many combinatorialists do, Stanley has phrased his questions in an informal way without mentioning counting classes.

The class $\sharpP$ is the correct class for some purposes, but for others it is too large.
For example, the determinant of a skew-symmetric matrix with entries from $\{-1,0,1\}$ is always non-negative, but this quantity is trivially in $\sharpP$, because the determinant can be computed in polynomial time.
This gives no satisfying insight into whether or not this quantity has a combinatorial interpretation.
Smaller counting classes are required (see also the discussion in \cite[\S1]{Pak18}), and we provide the first study of functional closure properties for the subclass $\sharpFA \subset \sharpP$.
Unlike the classification for $\sharpP$,
our results do not rely on oracle separations, i.e., our classification is entirely unconditional.

\subsection{Our results}
\label{subsec:ourresults}
Let $n\rem p \in\{0,\ldots,p-1\}$ denote the smallest nonnegative $r$ such that $n\equiv_p r$.
A function $\varphi: \IN \to \IN$ is called
\emph{ultimately PORC} (Polynomial On Residue Classes\footnote{PORC functions are also known as quasipolynomials or pseudopolynomials, but we want to avoid those names for the potential confusion to quasipolynomial growth and pseudopolynomial running times.}) if $\exists p,N\in\IN$ and there exist polynomials $\varphi_{0}, \ldots, \varphi_{p-1}: \IN \to \IQ$ such that for every $n \geq N$ we have
$\varphi(n) = \varphi_{n\rem p}(n)$ \ \footnote{Note that each $\varphi$ in this paper is defined on the natural numbers and maps to the natural numbers, which
is a subtle restriction. For example, a univariate polynomial $\varphi : \IQ \to \IQ$ maps integers to integers if and only if its coefficients in the binomial basis are integers, see Section~\ref{sec:notation}.
However, non-negativity is not an algebraic property.
Also note that for the case of $\varphi$ being just a univariate polynomial, the corresponding linear recursive sequence can have negative entries in the matrix.}.

\noindent We first classify the univariate functional closure properties of $\sharpFA$:
\begin{itemize}
\item\textsf{\textbf{Theorem}} (see Theorem~\ref{thm:univclosure}).
\textit{A function $\varphi: \IN \to \IN$ is a functional closure property of $\sharpFA$ if and only if $\varphi$ is an ultimately PORC function.}
\end{itemize}

\vspace{-0.3cm}

\noindent More generally, we classify the multivariate functional closure properties of $\sharpFA$:
\begin{itemize}
\item
\textsf{\textbf{Theorem}} (see Theorem~\ref{thm:multivclosure}).
\textit{A function $\varphi: \IN^m \to \IN$ is a functional closure property of $\sharpFA$ if and only if $\varphi$ can be written as a finite sum of finite products of univariate ultimately PORC functions.}
\end{itemize}

\vspace{-0.3cm}

\noindent We analyze the special case of multivariate polynomials:
\begin{itemize}
\item
\textsf{\textbf{Theorem}} (see Lemma~\ref{lem:multivarpolyn}).
\textit{A multivariate polynomial $\varphi: \IN^m \to \IN$ with rational coefficients is a functional closure property of $\sharpFA$ iff for every $\psi$ that can be formed from $\varphi$ by replacing any subset of variables -- including the empty set -- by constants from $\IN$, then all dominating terms of $\psi$ in the binomial basis have positive coefficients.}
\end{itemize}
\vspace{-0.3cm}

\noindent We lift this result to monotone graph varieties
(and to more general sets, see Theorem~\ref{thm:polyclusterseq}), where we get exactly the desirable classification given by the vanishing ideal:
\begin{itemize}
\item
\textsf{\textbf{Theorem}} (see Theorem~\ref{thm:mongraphclosures}).
\textit{Let $S$ be a monotone graph variety and let $I=I(S)$ be its vanishing ideal.
A multivariate polynomial $\varphi:S\to\IN$ is a functional promise closure property of $\sharpFA$ with regard to $S$ if and only if there exists $\psi\in I$ such that $\varphi+\psi$ is a multivariate functional closure property of $\sharpFA$.
}
\end{itemize}

\section{Notation}
\label{sec:notation}

Let $\IN=\{0,1,2,\ldots\}$. For a finite set $\Sigma$ let $\Sigma^\star$ denote the set of all finite length sequences with elements from~$\Sigma$.
The vector space of multivariate polynomials $\IQ[f_1,\ldots,f_m]$ in variables $f_1,\ldots,f_m$ has a basis given by products of binomial coefficients: $\big\{\prod_{i=1}^m \binom{f_i}{c_i}\big\}_{c_1,\ldots,c_m}$, where each $c_i\in\IN$. Here
we used
$\binom{x}{c} = \frac{1}{c!} x \cdot (x-1) \cdot \ldots \cdot (x-c+1)$ as a polynomial.
This is called the \emph{binomial basis}.
A multivariate polynomial $\varphi$ is called \emph{integer valued} if $\varphi(\IZ^m)\subseteq \IZ$, which is equivalent to $\varphi(\IN^m)\subseteq \IZ$, and which is also equivalent to all coefficients in the binomial basis being integers, see for example \cite[Prop.\ 4.2.1]{IP22} for a short proof of this classical fact.

We now recall (see \cite[Def.~2.1]{DK21}) our main model of computation, the finite $\IN$-weighted automaton, which we just call non-deterministic finite automaton (NFA) for brevity.

\vspace{-0.2cm}

\begin{definition}%
\label{def:NFA}
An NFA $M$ is a tuple $(Q, \Sigma, \wt, \inp, \outp)$ where the set of states $Q$ and the alphabet $\Sigma$ are finite sets and $\wt: Q \times \Sigma \times Q \to \IN$ is the weighted transition function, $\inp: Q \to \IN$ are the weighted initial states and $\outp: Q \to \IN$ are the weighted accepting states\footnote{Note that in definitions by other authors one can find simpler versions of NFAs, in particular unweighted initial and accepting states and unweighted edges while also restricting to a single initial state.
We will see soon that working with unweighted NFAs is not a restriction, but additionally restricting the model to have a single initial state is strictly weaker, since this model could not compute any function $f$ with $f(\eps) > 1$.
To obtain the same expressiveness one would have to additionally allow for $\eps$-transitions, while disallowing cycles of $\eps$-transitions to prevent infinite values for $f$.
}.
    A \emph{computation} $P$ for a word $w = w_1\ldots w_n \in \Sigma^\star$ of length $n$ is a sequence $q_0q_1 \ldots q_n$ in $Q^{n+1}$.
    It has \emph{multiplicity} or \emph{weight}\footnote{We will use both of these terms interchangeably. For a weighted automaton, calling this weight is more natural, while when looking at the underlying graph as a multigraph, multiplicity of paths and walks is more natural.} $\Cweight(P) = \inp(q_0) \cdot \prod_{i=1}^n \wt(q_{i-1}, w_i, q_i) \cdot \outp(q_n)$ and \emph{partial weight} $\Cpweight(P) = \inp(q_0) \cdot \prod_{i=1}^n \wt(q_{i-1}, w_i, q_i)$.
    We say that $M$ \emph{computes} $f: \Sigma^\star \to \IN$ where $f(w)$ is the sum of the weights over all computations of $M$ on $w$.
    The class $\sharpFA$ is defined as the set of all functions $f: \Sigma^\star \to \IN$ that are computed by NFAs.
\end{definition}

\vspace{-0.2cm}

If needed to distinguish these for different automata, we use a corresponding subscript, for example the weights of computations in $M_f$ would be denoted by $\Cweight_f$, etc.

\vspace{-0.2cm}

\begin{definition}[Simple NFA]
    We say an NFA $M = (Q, \Sigma, \wt, \inp, \outp)$ is \emph{simple} if $\image \wt, \image \inp, \image \outp \subseteq \{0, 1\}$.
\end{definition}

The notion of a simple NFA also motivates our use of the term NFA opposed to $\IN$-weighted automaton: We simply count the number of accepting paths of $M$ on a word $w$.
This is in line with $\sharpP$ counting the number of accepting paths on a polynomial time non-deterministic Turing machine.

\begin{lemma}[Folklore]
    \label{lem:simplify}
    For every NFA $M$ there exists a simple NFA $M'$ computing the same function.%
\end{lemma}

\vspace{-0.2cm}

A proof of this simple fact can be found in the Appendix, for the sake of completeness.
We denote by $a \equiv_p b$ that $a \in \IN$ and $b \in \IN$ are congruent modulo $p \in \IN \setminus \{0\}$.
The indicator function $\ONE_{n = c}: \IN \to \IN$ is defined as $n \mapsto \begin{cases}1 & \text{if $n = c$} \\ 0 & \text{otherwise}\end{cases}$ and analogously for different conditions.
By abuse of notation, if we have a function $f: \Sigma^\star \to \IN$ and an expression $\varphi: \IN \to \IN$ in $n$, we replace $n$ by $f$ in the expression to denote $\varphi \circ f$.
For example we use $\ONE_{f = c}$ to denote the function $w \mapsto \ONE_{f(w) = c}$, similarly $\binom{f}{2}$ denotes the function $w \mapsto \binom{f(w)}{2}$.
Furthermore we use the notation $[n]$ to denote the set $\{1, \ldots, n\}$ for any $n \in \IN$.

\section{Functional closure properties}
\label{sec:functional_closure_properties}
\subsection{Univariate functional closure properties}
We say a function $\varphi: \IN \to \IN$ is a functional closure property of $\sharpFA$ if $\varphi(\sharpFA) \subseteq \sharpFA$, i.e.\ if for every function $f \in \sharpFA$ the function $\varphi \circ f$ is also in $\sharpFA$.
Our goal in this section is to classify all functional closure properties of $\sharpFA$.
They will be precisely the ultimately PORC functions.

We call a function $\varphi: \IN \to \IN$ an \emph{ultimately almost PORC function%
} if there is a \emph{quasiperiod}~$p$, an \emph{offset} $N \in \IN$ and \emph{constituents} $\varphi_{0}, \ldots, \varphi_{p-1}: \IN \to \IQ$, where each $\varphi_i$ is either a polynomial with rational coefficients or a function in $2^{\Theta(n)}$, and for every $n \geq N$ we have
$\varphi(n) = \varphi_{n\rem p}(n)$.
If all the constituents are polynomials, we call $\varphi$ an \emph{ultimately PORC function}.
The smallest representative of each constituent and of the finite cases before the periodic behaviour is captured by the \emph{shifted remainder} operator $\smod{n}{N}{p}$, defined via

\vspace{-0.6cm}

\[
    \smod{n}{N}{p} = \begin{cases}
        n & \text{if $n < N$}\\
        \min\{k \geq N \mid k \equiv_p n\} & \text{if $n \geq N$}
    \end{cases}
\]

\vspace{-0.2cm}

The first half of the section is dedicated to proving that every ultimately PORC function is a functional closure property of $\sharpFA$, see Lemma~\ref{lem:univclosure:upper_bound}.
In order to prove this we show that $\sharpFA$ is closed under
\begin{restatable}[Addition]{lemma}{restateadd}
    \label{lem:add}
    If $f, g \in \sharpFA$, then $f+g \in \sharpFA$.
\end{restatable}
\vspace{-.5cm}
\begin{restatable}[Multiplication]{lemma}{restatemult}
    \label{lem:mult}
    If $f, g \in \sharpFA$, then $f \cdot g \in \sharpFA$.
\end{restatable}
\vspace{-.5cm}
\begin{restatable*}[Subtraction of constants]{lemma}{restatesubtraction}
    \label{lem:subtraction}
    If $f \in \sharpFA$, then $\forall c \in \IN: \, \max(f-c, 0) \in \sharpFA$.
\end{restatable*}
\vspace{-.5cm}
\begin{restatable*}[Clamping]{lemma}{restateclamping}
    \label{lem:clamping}
    If $f \in \sharpFA$, then $\min(f, c) \in \sharpFA$ for any constant $c \in \IN$.
\end{restatable*}
\vspace{-.5cm}
\begin{restatable*}[Comparison with constants]{lemma}{restateconst}
    \label{lem:const}
    If $f \in \sharpFA$, then the functions $\ONE_{f = c}$, $\ONE_{f \leq c}$, $\ONE_{f \geq c}$ are in $\sharpFA$ for any constant $c \in \IN$.
\end{restatable*}
\vspace{-.5cm}
\begin{restatable*}[Division by constants]{lemma}{restatedivision}
    \label{lem:division}
    If $f \in \sharpFA$, then $\forall c \in \IN \setminus \{0\} : \, \lfloor f / c \rfloor \in \sharpFA$.
\end{restatable*}
\vspace{-.5cm}
\begin{restatable*}[Modular arithmetic]{lemma}{restatemod}
    \label{lem:mod}
    If $f \in \sharpFA$, then the function $\ONE_{f \equiv_c d}$ is in $\sharpFA$ for any constants $c \in \IN \setminus \{0\}$ and $d \in \IZ_c$.
\end{restatable*}
\vspace{-.5cm}
\begin{restatable*}[Binomial coefficients]{lemma}{restatebinom}
    \label{lem:binom}
    If $f \in \sharpFA$, then $\binom{f}{c} \in \sharpFA$ for any constant $c \in \IN$.
\end{restatable*}

While addition and multiplication are technically bivariate functional closure properties, we list them here already since they are abundantly used throughout the proofs of the univariate functional closure properties.
Proofs of those two classical results can be found in \cite{droste2009handbook}[Ch. 4.1 and 4.2.2] and in the appendix, for the sake of completeness.

With the exception of binomial coefficients all the other closure properties need to be able to ``remove'' some of the possible computations.
For example, consider decrementation, the special case of truncated subtraction by one, and consider some simple NFA $M$ computing some strictly positive function $f$.
We now want to construct an NFA $M'$ that computes $f-1$, i.e. an NFA that has exactly one non-zero computation less than $M$ (assuming computations of weights zero or one).
For this we want a procedure to single out one non-zero computation of $M$ to then change its weight to zero.
For stronger models of computation -- like polynomial time non-deterministic Turing machines -- this approach seems hopeless.
Already deciding the existence of one such computation is $\NP$-hard.
However for NFAs, deciding the existence of a non-zero computation can be decided by a deterministic finite automaton, namely the powerset automaton.
Adjusting the powerset construction to filter out a single non-zero computation, namely the lexicographically minimal one can then be used to show that decrementation is a closure property of $\sharpFA$.

Generalizing this approach to more general properties about the computations gives us the framework of stepwise computation properties:

\vspace{-0.2cm}

\begin{definition}[Stepwise computation property]
    Let $M = (Q, \Sigma, \wt, \inp, \outp)$ be an NFA.
    A \emph{stepwise computation property} $\prop$ is defined as $\prop = (S, \init, \step, \cond)$ where $S$ is a finite set and $\init: Q \to S$, $\step: Q \times \Sigma \times Q \times S \to S$ and $\cond: S \to \{0, 1\}$ are functions.
    For $w = w_1 \ldots w_n \in \Sigma^\star$ and a computation $P = q_0 \ldots q_n$ of $M$ on $w$ we define a \emph{step sequence} $s_0 := \init(q_0)$ and $s_i := \step(q_{i-1}, w_i, q_i, s_{i-1})$ for $i \in [n]$.
    We also write
$\prop(w, P) := \cond(s_n)$ to be the evaluation of the property.
\end{definition}

These stepwise computation properties now enable us, given a simple NFA, to construct NFAs computing both of the following:

\vspace{-0.2cm}

\begin{restatable}{lemma}{restatebaseconstruction}
    \label{lem:baseconstruction}
    Let $M_f$ %
    be a simple NFA computing a function $f$ and let $\prop$ %
    be a stepwise computation property.
    Then there is an NFA $M$ %
    computing $g(w) = \sum_{P} \Cweight_f(P) \cdot \prop(w, P)$, where the sum is over all computations $P$ of $M_f$ on $w$.
\end{restatable}
\vspace{-0.4cm}
\begin{restatable}{lemma}{restatebaseconstructionall}
    \label{lem:baseconstruction_all}
    Let $M_f$ %
    be a simple NFA computing a function $f$ and let $\prop$ %
    be a stepwise computation property.
    Then there is an NFA $M$ %
    computing $g(w) = \sum_{P} \prop(w, P)$, where the sum is over all computations $P$ of $M_f$ on $w$.
\end{restatable}

\vspace{-0.5cm}

\begin{proofsketch}
    For both of these lemmas, we construct a sort of product automaton of $M_f$ and $\prop$ (represented by the set $S$), the details can be found in the appendix. 
\end{proofsketch}

\vspace{-0.2cm}

In other words, stepwise computation properties allow us to either ``disable'' specific computations of $M_f$ or they allow us to directly extract information about the computations of $M_f$.
Note that the restriction on the finiteness of $S$ is necessary, as the elements of $S$ are hard-coded into the state space of the NFA in Lemmas~\ref{lem:baseconstruction} and~\ref{lem:baseconstruction_all}.
In particular, these lemmas do not hold for even countably infinite $S$, which can be seen with the example $S=\IN$ when we define the stepwise computation property in such a way that $\prop(w, P) = 1$ iff $|w|$ is prime, leading to an NFA recognizing the language of all words of prime length, a well known contradiction.

Further note that these two constructions do not incur exponential blowups themselves, however for most of our applications the set $S$ will be of exponential size in the number of states of $M_f$.

Returning to our decrementation example, to use Lemma~\ref{lem:baseconstruction} we want to construct a stepwise computation property $\prop$ with $\prop(w, P) = 0$ iff $P$ is the lexicographically smallest non-zero weight computation on $w$.
For this we can set $S = \powset(Q)$ to be the set of all subsets of states $Q$, denoting the set of states that currently are the endpoints of lexicographically smaller partial computations of non-zero weight than the partial computation $P$ we are on.
We need to store all such potential states, since some current partial non-zero weight computations might not be possible to be completed to a full non-zero weight computation.
Initially this set contains all states that are smaller than the start state of $P$, the $\step$ function then checks which lexicographically smaller partial computations can be extended and whether any new partial computations that agreed with $P$ up to this state can be lexicographically smaller than $P$.
Finally the $\cond$ function then checks whether there are any such lexicographically smaller partial computations left that can be completed to a non-zero weight computation, i.e.\ that end on a state $q \in Q$ with $\outp(q) = 1$.

We want to generalize this idea to be able to generally create stepwise computation properties that argue about the number of non-zero computations, either in total or lexicographically smaller than a given computation.
However doing this in general would require choosing the set $S$ as the set of functions $Q \to \IN$, which is infinite.
As a result we embed the number of computations into finite semirings first to then extract the relevant information.
For this purpose we only consider semirings with both additive and multiplicative identities.
Homomorphisms $h$ from a semiring $\CR$ into another semiring $\CR'$ need to fulfill
$h(a+b) = h(a) + h(b)$, \,
$h(a \cdot b) = h(a) \cdot h(b)$, \,
$h(1_{\CR}) = 1_{\CR'}$, \,
$h(0_{\CR}) = 0_{\CR'}$.
Since every element of $\IN$ is either $0$ or can be formed by repeated addition of $1$, any homomorphism from $\IN$ into any other semiring is uniquely defined.

We can then use $\CR$ to construct stepwise computation properties (the full proofs can be found in the appendix).
Combined with Lemmas~\ref{lem:baseconstruction} and~\ref{lem:baseconstruction_all} these use similar ideas to \cite{klimann2004deciding}.

\vspace{-0.2cm}

\begin{restatable}{lemma}{restatepropall}
    \label{lem:prop_all}
    Let $N = (Q, \Sigma, \wt, \inp, \outp)$ be a simple NFA and let $\CR$ be a finite semiring and let $\tau: \IN \to \CR$ be the unique homomorphism from $\IN$ to $\CR$.
    For any function $\pi: \CR \to \{0, 1\}$ there is a stepwise computation property $\prop$ with $\prop(w, P) = \pi(\tau(\sum_{P'}\Cweight(P')))$, where the sum is over all computations $P'$ of $N$ on $w$, independent of $P$.
\end{restatable}
\vspace{-0.4cm}
\begin{proofsketch}
    Construct $\prop = (S, \init, \step, \cond)$ via:
$S = Q \to \CR$, \ $\init(q) = r \mapsto \tau(\inp(r))$, \ 
$\step(q, \sigma, q', s) = r \mapsto \textstyle\sum_{r' \in Q}s(r') \cdot \tau(\wt(r', \sigma, r))$,
and $\cond(s) = \pi\bigl(\textstyle\sum_{r \in Q}s(r) \cdot \tau(\outp(r))\bigr)
$.    
    Let $s_0, \ldots, s_n$ be the step sequence of any computation $P$ of $w$.
    Since $\tau$ is a homomorphism, we can pull out $\tau$.
    Thus $s_i(q) = \tau(\sum_{\tilde P}\Cpweight(\tilde P))$, where the sum is over all computations $\tilde P$ of $w_1, \ldots w_i$ ending in the state $q$.
    The condition $\cond$ then completes this to $\prop(w, P) = \pi(\tau(\sum_{P'}\Cweight(P')))$, where the sum is over all computations $P'$ of $N$ on $w$.
\end{proofsketch}

\vspace{-0.4cm}

\begin{restatable}{lemma}{restateproplexi}
    \label{lem:prop_lexi}
    Let $M = (Q, \Sigma, \wt, \inp, \outp)$ be a simple NFA with some ordering $<$ of $Q$, let $\CR$ be a finite semiring and let $\tau: \IN \to \CR$ be the unique homomorphism from $\IN$ to $\CR$.
    For any function $\pi: \CR \to \{0, 1\}$ there is a stepwise computation property $\prop$ with $\prop(w, P) = \pi(\tau(\sum_{P'}\Cweight(P')))$ for all non-zero computations $P$ of $N$ on $w$, where the sum is over all computations $P'$ of $N$ on $w$ that are lexicographically smaller than $P$.
\end{restatable}
\vspace{-0.4cm}
\begin{proofsketch}Construct $\prop = (S, \init, \step, \cond)$ via:
$S = Q \to \CR$, \ $\init(q) = r \mapsto \tau(\ONE_{r < q} \cdot \inp(r))$, \  
$\step(q, \sigma, q', s) = r \mapsto \textstyle\sum_{r' \in Q}s(r') \cdot \tau(\wt(r', \sigma, r)) + \tau(\ONE_{r < q'} \cdot \wt(q, \sigma, r))$, \ 
$\cond(s) = \pi(\textstyle\sum_{r \in Q}\tau(\outp(r)) \cdot s(r))$.
    Let $s_0, \ldots, s_n$ be the step sequence of any computation $P = q_0 \ldots q_n$ of $w$.
    Inductively we can show that $s_{i}(r) = \sum_{P'} \tau(\Cpweight(P'))$ for all $i \in \{0, \ldots, n\}$ and $r \in Q$, where the sum is over all computations $P' = q_0' \ldots q_{i}'$ of $N$ on $w_1 \ldots w_{i}$ with $q_{i}' = r$ that are lexicographically smaller than $q_0 \ldots q_{i}$.
    Initially the only computations $P' = q_0'$ that are lexicographically smaller than the computation $q_0$ are the ones with $q_0' < q_0$.
    For $i > 0$ for a computation $P' = q_0' \ldots q_{i}'$ to be lexicographically smaller than $q_0 \ldots q_{i}$ there are two possibilities.
    Either $q_0' \ldots q_{i-1}'$ is already lexicographically smaller than $q_0 \ldots q_{i-1}$  or $q_0' \ldots q_{i-1}' = q_0 \ldots q_{i-1}$ and $q_{i}' < q_{i}$.
    In the second case the weight of $P'$ is precisely $w(q_{i-1}, w_{i}, q_{i}')$ since $P$ is a computation of non-zero weight and thus weight exactly $1$.
    Combining all of this, we can finish the proof of the claim with a similar argument to Lemma~\ref{lem:prop_all}.
\end{proofsketch}

\vspace{-0.2cm}

Note that the previous lemma makes no statement about the value of $\prop(w, P)$ for any computations $P$ of weight zero.
However, this is enough for our uses, since we only combine it with Lemma~\ref{lem:baseconstruction}, i.e., $\prop(w, P)$ gets weighted by $\Cweight(P)$.

Most commonly, as is the case for decrementation, we want to be able to exactly distinguish the number of non-zero computations if it is less than $k$ and otherwise be able to tell that the number is at least $k$.
This is achieved by using the following capped semiring:

\vspace{-0.1cm}

\begin{definition}[Capped semiring]
    For $k \in \IN$ we call the semiring $\CR_k = \{0, \ldots, k\}$ with the operations $a +_{\CR} b := \min(a + b, k)$ and $a \cdot_{\CR} b := \min(a \cdot b, k)$ the \emph{capped} semiring.
\end{definition}

\vspace{-0.2cm}

We can now show that decrementation is a functional closure property by simply using Lemma~\ref{lem:prop_lexi} using the capped semiring $\CR_1$ and $\pi(a) = a$ to construct our wanted stepwise computation property computing $\prop(w, P) = 0$ iff $P$ is the lexicographically smallest non-zero weight computation on $w$.

Most of the remaining closure properties are now proven by using the capped semiring of a specific size and choosing the function $\pi$ accordingly, we will show this in detail for the example of subtraction, the remaining proofs can be found in the appendix.

\vspace{-0.1cm}

\restatesubtraction

\vspace{-0.4cm}

\begin{proof}
    Let $M_f = (Q_f, \Sigma, \wt_f, \inp_f, \outp_f)$ be a simple NFA computing $f$ with an arbitrary ordering $<$ on $Q_f$.
    Lemma~\ref{lem:prop_lexi} on the capped semiring $\CR_{c}$ with $\pi(a) = \ONE_{a \geq c}$ for all $a \in \CR$ constructs a stepwise computation property $\prop$ with
    \vspace{-0.2cm}
    \[
        \prop(w, P) = \begin{cases}
            1 & \begin{minipage}{11cm}
            if the number of non-zero computations $P'$ on $w$ that are lex. smaller\\[-3pt] than $P$ is at least $c$
            \end{minipage}\\[3pt]
            0 & \text{otherwise}\\
        \end{cases}
    \]

    \vspace{-0.2cm}

    for all computations $P$ on $w$ of non-zero weight, i.e., $\prop(w, P) = 0$ iff $P$ is one of the $c$ lexicographically smallest computations on $w$ with non-zero weight.
    It follows that the NFA $M$ constructed by Lemma~\ref{lem:baseconstruction} computes $g(w) = \max(f(w) - c, 0)$.
\end{proof}

\vspace{-0.2cm}

If instead of rejecting the $c$ lexicographically smallest computations, we accept only those computations, we compute the minimum of $f$ and $c$.

\vspace{-0.1cm}

\restateclamping

\vspace{-0.2cm}

By using the capped semiring $\CR_{c+1}$ with $\pi_{=}(a) = \ONE_{a = c}$, $\pi_{\leq}(a) = \ONE_{a \leq c}$ and $\pi_{\geq}(a) = \ONE_{a \geq c}$, we can compute the indicator functions $\ONE_{f = c}$, $\ONE_{f \leq c}$ and $\ONE_{f \geq c}$ respectively.

\vspace{-0.1cm}

\restateconst

\vspace{-0.2cm}

The previous lemma in particular also implies the following:

\vspace{-0.2cm}

\begin{lemma}
    \label{lem:finitedifference}
    If $\varphi: \IN \to \IN$ is a functional closure property of $\sharpFA$ and $\psi: \IN \to \IN$ is an arbitrary function with $\varphi(n) = \psi(n)$ for all but finitely many $n \in \IN$, then $\psi$ is also a functional closure property of $\sharpFA$.
\end{lemma}

\vspace{-0.4cm}

\begin{proof}
    Let $f \in \sharpFA$ be arbitrary.
    Further let $N \in \IN$ be such that $\varphi(n) = \psi(n)$ for all $n \geq N$.
    Then $\textstyle (\psi \circ f)(w) = \sum_{i=0}^{N-1}\ONE_{f(w)=i}\psi(i) + \ONE_{f(w) \geq N} \varphi(f(w))$
    for all $w \in \Sigma^\star$.
    In particular $\psi \circ f \in \sharpFA$ since the $\psi(i)$ are constants and $\varphi \circ f \in \sharpFA$.
\end{proof}

\vspace{-0.2cm}

Division and modular arithmetic however use a different semiring: they use the finite cyclic semiring $\IZ_c = \{0, \ldots, c-1\}$ with $\pi(a) = \ONE_{a = c-1}$ and $\pi(a) = \ONE_{a = d}$ respectively.

\vspace{-0.1cm}

\restatedivision
\vspace{-0.3cm}
\restatemod

\vspace{-0.2cm}

The previous closure properties turn out to already be sufficient to generate all functional closure properties, so in particular they are sufficient to generate binomial coefficients by using subtraction of constants, multiplication and division by constants by using the definition of binomial coefficients as a polynomial: \mbox{$\binom{x}{c} = \frac{1}{c!} x \cdot (x-1) \cdot \ldots \cdot (x-c+1)$}.
Nonetheless, we give an additional proof for binomial coefficients as a different interesting application of the stepwise computation property framework.

\restatebinom

\vspace{-0.4cm}

\begin{proof}
    Let $M_f = (Q_f, \Sigma, \wt_f, \inp_f, \outp_f)$ be a simple NFA computing $f$ and let $c \in \IN$.
    For $c < 2$ the statement of this lemma is trivially true, so assume $c \geq 2$.

    We construct the $c$-fold product automaton $M_f^c = (Q_f^c, \Sigma, \wt_f^c, \inp_f^c, \outp_f^c)$ with

\vspace{-.7cm}
    
    \begin{align*}
        \wt_f^c((q_1, \ldots, q_c), \sigma, (q_1', \ldots, q_c')) &= \textstyle\prod_{i=1}^c \wt_f(q_i, \sigma, q_i')\\
        \inp_f^c((q_1, \ldots, q_c)) &= \textstyle\prod_{i=1}^c \inp_f(q_i)\\
        \outp_f^c((q_1, \ldots, q_c)) &= \textstyle\prod_{i=1}^c \outp_f(q_i)
    \end{align*}

\vspace{-.2cm}
    
    $M_f^c$ is a simple NFA and every computation on $M_f^c$ is the cartesian product of $c$ computations on $M_f$.
    Our aim is to now construct a stepwise computation property $\prop = (S, \init, \step, \cond)$ such that $\prop(w, P) = 1$ iff $P$ is composed of $c$ pairwise distinct computations\footnote{We could also require them to be sorted in lexicographical order by having $S$ be the set of all total preorders, but since we can divide by $c!$ we are going with the easier exposition.} on $N_f$.

    For this let $S$ be the set of all equivalence relations on the set $[c]$.
    We define $\init((q_0, \ldots, q_c))$ to be the equivalence relation $R_0$ with $(a,b) \in R_0$ iff $q_a = q_b$.
    Additionally we define $\cond(R) = 1$ iff $R$ is the equivalence relation where every element is only equivalent to itself, i.e.\ a computation gets accepted iff all its constituent computations are pairwise distinct.
    Finally we define $\step((q_1, \ldots, q_c), \sigma, (q_1', \ldots, q_c'), R)$ to be the equivalence relation $R'$ defined via $(a, b) \in R'$ iff $(a, b) \in R$ and $q_a' = q_b'$.
    With a simple induction we can prove that for a computation $P = P_1 \times \ldots \times P_c$ and the step sequence $R_0, \ldots, R_n$ we have $(a, b) \in R_i$ iff the computations $P_a$ and $P_b$ are identical for the first $i$ steps.
    It follows that the NFA $M$ constructed by Lemma~\ref{lem:baseconstruction} computes $g(w) = \binom{f(w)}{c} \cdot c!$.
    Now, $\binom{f}{c} \in \sharpFA$ 
    by Lemma~\ref{lem:division}.
\end{proof}

\vspace{-0.2cm}

While the combination of the previous lemmas can be used to show that any polynomial written in the binomial basis with non-negative integer coefficients is a functional closure property of $\sharpFA$, we can do better by considering a shifted binomial basis.
For example, consider the polynomial $\varphi(x) = \frac{x^2}{2} - \frac{3x}{2} + 1$.
This polynomial is non-negative for all $x \in \IN$.
Writing $\varphi$ in the binomial basis we get $\varphi(x) = \binom{x}{2} - \binom{x}{1} + 1$.
If however we allow the upper indices of the binomial basis to be shifted, we can write $\varphi$ without the use of negative coefficients as $\varphi(x) = \binom{x-1}{2}$.
While $x-1$ itself is not a functional closure property of $\sharpFA$ the function $\max(x-1, 0)$ is a functional closure property of $\sharpFA$ and is different from $x-1$ for only finitely many $x \in \IN$.
In the same way we see that $\varphi'(x) := \binom{\max(x-1, 0)}{2}$ only differs from $\varphi$ for finitely many $x \in \IN$, namely $x = 0$.
Using Lemma~\ref{lem:finitedifference} to change those finitely many values, we see that $\varphi$ is indeed a functional closure property of $\sharpFA$.

Generalizing this idea we will show with the next two lemmas that this is possible for any $\varphi$ with integer coefficients in the binomial basis, with a small restriction: We don't show that $\varphi$ itself is a functional closure property of $\sharpFA$, but rather that $x \mapsto \max(\varphi(x), 0)$ is one.
Note that this restriction is the best we can hope for, since no computation in an NFA can ever have negative weight.

\vspace{-0.2cm}

\begin{restatable}{lemma}{restatebinbasischange}
    \label{lem:binbasischange}
    Let $\varphi(x) = \sum_{i=0}^r a_i \cdot \binom{x}{i}$ with $a_i \in \IZ$ and $a_r > 0$.
    Then there are $b_0, \ldots, b_r \in \IN$ and $c_0, \ldots, c_r \in \IN$ with $\varphi(x) = \sum_{i=0}^r b_i \cdot \binom{x-c_i}{i}$.
\end{restatable}

\vspace{-0.5cm}

\begin{proofsketch}
    We inductively prove this claim by using the Chu-Vandermonde identity \cite{Spi16} on the term of highest degree.
    It allows us to replace the highest degree binomial via $\binom{x-c_r}{r} = \sum_{i=0}^r(-1)^{r-i} \binom{r-i+c_r-1}{r-i} \binom{x}{i}$.
    For sufficiently large $c_r \in \IN$ this implies that the leading term of $\varphi(x) - a_r \cdot \binom{x-c_r}{r}$ again is positive and of smaller degree.
\end{proofsketch}

\vspace{-0.2cm}

A full proof of the previous lemma can be found in the appendix.

\begin{lemma}[Integer-valued polynomials]
    \label{lem:polynomial}
    Let $f \in \sharpFA$ and let $\varphi: \IQ \to \IQ$ be an integer-valued polynomial, then $\max(\varphi \circ f, 0) \in \sharpFA$.
\end{lemma}

\vspace{-0.5cm}

\begin{proof}
    We can assume the leading coefficient of $\varphi$ to be positive.
    Otherwise $\max(\varphi \circ f, 0)$ can be directly written as a finite sum $\sum_{i}c_i \cdot \ONE_{f=i}$ which is in $\sharpFA$ by Lemmas~\ref{lem:add}, \ref{lem:mult} and \ref{lem:const}.
    Write $\varphi$ in the binomial basis as $\varphi(x) = a_0 \cdot \binom{x}{0} + \ldots + a_r \cdot \binom{x}{r}$ with $a_r > 0$.
    Since $\varphi$ is integer-valued, all of the $a_i$ are integers, see \cite[Prop.~4.2.1]{IP22}.
    Using Lemma~\ref{lem:binbasischange} we get a representation $\varphi(x) = \sum_{i=0}^r b_i \cdot \binom{x-c_i}{i}$ with $b_0, \ldots, b_r \in \IN$ and $c_0, \ldots, c_r \in \IN$.
    For $x \geq \max\{c_i \mid 0 \leq i \leq r\} =: N$ we have that $\binom{x-c_i}{i} = \binom{\max(x - c_i, 0)}{i}$ and thus $\psi(x) := \sum_{i=0}^r b_i \cdot \binom{\max(x - c_i, 0)}{i}$ only differs from $x \mapsto \max(\varphi(x), 0)$ on finitely many inputs and is a functional closure property of $\sharpFA$ by Lemmas~\ref{lem:add},~\ref{lem:mult},~\ref{lem:subtraction} and~\ref{lem:binom}.
    Lemma~\ref{lem:finitedifference} then finishes off the claim.
\end{proof}

We now have the tools available to show our claim that every ultimately PORC function is a functional closure property of $\sharpFA$.

\vspace{-0.2cm}

\begin{lemma}
    \label{lem:univclosure:upper_bound}
    Every ultimately PORC function is a functional closure property of $\sharpFA$.
\end{lemma}

\vspace{-0.5cm}

\begin{proof}
    Let $\varphi: \IN \to \IN$ be an ultimately PORC function with period $p$ comprised of the polynomial constituents $\varphi_0, \ldots, \varphi_{p-1}: \IN \to \IQ$ and $N \in \IN$, s.t.\ for all $n \geq N$ we have
    $\varphi(n) = \varphi_{n\rem p}(n)$.
    Additionally let $f \in \sharpFA$.
    We can write
    
    \vspace{-0.7cm}
    \[
    \textstyle
        \varphi \circ f = \sum_{i=0}^{N-1}\ONE_{f=i} \cdot \varphi(i) + \ONE_{f \geq N} \cdot (\sum_{i=0}^{p-1}\ONE_{f \equiv_p i} \cdot \lfloor \max(\varphi_i \circ f, 0) \rfloor)\,.
    \]
    \vspace{-0.7cm}
    
    Combining Lemmas~\ref{lem:add},~\ref{lem:mult},~\ref{lem:const} and~\ref{lem:mod} this shows $\varphi \circ f \in \sharpFA$, if we can show $\lfloor \max(\varphi_i \circ f, 0) \rfloor \in \sharpFA$ for all $i \in \{0, \ldots, p-1\}$.
    To show this let $\alpha_i$ be the common denominator of the coefficients of $\varphi_i$.
    Then $\alpha_i \cdot \varphi_i$ is a polynomial with integer coefficients, so it in particular is an integer-valued polynomial and by Lemma~\ref{lem:polynomial} we have that $\max(\alpha_i \cdot \varphi_i \circ f, 0) \in \sharpFA$.
    Combining this with Lemma~\ref{lem:division} we get that $\left\lfloor \frac{\max(\alpha_i \cdot \varphi_i \circ f, 0)}{\alpha_i} \right\rfloor = \lfloor \max(\varphi_i \circ f, 0) \rfloor \in \sharpFA$.
\end{proof}

The remainder of this section is dedicated to showing that no other functional closure properties of $\sharpFA$ exist.
This will make use of the following well known algebraic interpretation of NFAs:

\vspace{-0.2cm}

\begin{lemma}[see \cite{Sch61}]
    \label{lem:algebraic}
    If $M = (Q, \Sigma, \wt, \inp, \outp)$ is an NFA, then there are matrices $A_\sigma \in \IN^{|Q| \times |Q|}$ for each symbol $\sigma \in \Sigma$ and vectors $a, b \in \IN^{|Q|}$, s.t. $M$ computes $a^T \cdot \left(\prod_{j=1}^{|w|}A_{w_j}\right) \cdot b$ for all $w \in \Sigma^\star$.
\end{lemma}

\vspace{-0.5cm}

\begin{proof}
    We index $A_\sigma$, $a$ and $b$ using states $q, q' \in Q$.
    Choose $(A_\sigma)_{q, q'} = \wt(q, \sigma, q')$, $a_q = \inp(q)$ and $b_q = \outp(q)$.
    It is now easy to see that $M$ computes exactly $a^T \cdot \left(\prod_{j=1}^{|w|}A_{w_j}\right) \cdot b$.
\end{proof}

\vspace{-0.2cm}

When restricting to a unary alphabet $\Sigma = \{\sigma\}$, this degenerates the computed function to $a^T \cdot A_{\sigma}^{|w|} \cdot b$.
In order to analyze the behaviour of these functions we first analyze the behaviour of the matrix power as a function in $|w|$ in the next two lemmas.
Their proofs can be found in the appendix.

\begin{restatable}{lemma}{restatediagonalstructure}
    \label{lem:diagonalstructure}
    Let $A \in \IN^{k \times k}$.
    Then any diagonal entry of $A^n$ is a function $f: \IN \to \IN$ with one of the following properties:
    \begin{enumerate}
        \item $f(n) = 0$ for all $n \in \IN \setminus \{0\}$ and $f(0) = 1$.
        \item There is a $p \in \IN \setminus \{0\}$, such that for all $n \in \IN$ we have $f(n) = \ONE_{n \equiv_p 0}$.
        \item There is a $p \in \IN \setminus \{0\}$ and a function $g \in 2^{\Theta(n)}$, such that for all $n \in \IN$ we have $f(n) = \ONE_{n \equiv_p 0} \cdot g(n)$.
    \end{enumerate}
\end{restatable}
These naturally correspond to vertices $v$ in the multigraph defined by the adjacency matrix $A$ with
\begin{enumerate}
    \item no paths from $v$ to $v$.
    \item exactly one path from $v$ to $v$ of length $p$.
    \item multiple walks from $v$ to $v$ where the lengths of all the walks from $v$ to $v$ have gcd $p$.
\end{enumerate}
We can then lift this result to all entries of $A^n$.

\vspace{-0.2cm}

\begin{restatable}{lemma}{restatepowerstructure}
    \label{lem:powerstructure}
    If $A \in \IN^{k \times k}$,
    then each entry of $A^n$ is an ultimately almost PORC function.
\end{restatable}

\vspace{-0.5cm}

\begin{theorem}[Classification of univariate functional closure properties of $\sharpFA$]
    \label{thm:univclosure}
    A function $\varphi: \IN \to \IN$ is a functional closure property of $\sharpFA$ iff $\varphi$ is an ultimately PORC function.
    This even holds when $\sharpFA$ is restricted to unary languages.
\end{theorem}

\vspace{-0.5cm}

\begin{proof}
    \begin{figure}
        \centering
        \begin{tikzpicture}
            \node[start] (S) {$S$};
            \node[accept, right of=S] (A) {$A$};

            \path
                (S) edge [loop above] node[left] {$1$ \ } (S)
                (A) edge [loop above] node[left] {$1$ \ } (A)

                (S) edge [-stealth] node {$1$} (A)
                ;
        \end{tikzpicture}\hspace{4cm}
        \begin{tikzpicture}
            \node[start, accept] (S) {$S$};

            \path
                (S) edge [loop above] node[left] {$1, 1$ \ } (S)
                ;
        \end{tikzpicture}

        \caption{NFAs computing the functions $1^n \mapsto n$ and $1^n \mapsto 2^n$ respectively. Edges with a multiplicity of $2$ are denoted by listing the edge label twice.}
        \label{fig:thm:univclosure:linear}
    \end{figure}
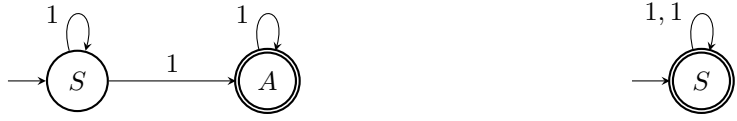
    Lemma~\ref{lem:univclosure:upper_bound} already shows that every ultimately PORC function is a closure property of $\sharpFA$.
    It remains to show that all functional closure properties of $\sharpFA$ are ultimately PORC functions.
    For this let $\varphi: \IN \to \IN$ be a functional closure property of $\sharpFA$.
    The function $f: \{1\}^\star \to \IN$ defined by $f(1^n) = n$ is computed by the left NFA in Figure~\ref{fig:thm:univclosure:linear} and thus in $\sharpFA$.
    Consequently $\varphi \circ f \in \sharpFA$.
    Let $M = (Q, \{1\}, \wt, \inp, \outp)$ be an NFA computing $\varphi \circ f$.
    By Lemma~\ref{lem:algebraic} this NFA induces a transition matrix $A \in \IN^{|Q| \times |Q|}$ and vectors $a, b \in \IN^{|Q|}$, s.t.\ $a^TA^nb = \varphi(f(1^n)) = \varphi(n)$ for all $n \in \IN$.
    Every entry of $A^n$ is an ultimately almost PORC function by Lemma~\ref{lem:powerstructure} and thus $a^TA^nb = \varphi(n)$ is one as well.
    Let $p$ be the quasiperiod of $\varphi$ and let $\varphi_i$ be one of the constituents of $\varphi$ corresponding to some residue class $i \in \{0, \ldots, p-1\}$.
    Assume for the sake of contradiction that $\varphi_i$ grows in $2^{\Theta(n)}$, i.e.\ there is a constant $\gamma \in \IR^+$ and $N \in \IN$, s.t.\ for every $n \geq N$ we have $\varphi_i(n) \geq 2^{\gamma n}$.
    Consider the function $f_{p,i}: \{1\}^\star \to \IN$ defined by $f_{p,i}(1^n) = p \cdot (2^n + N) + i$.
    We claim $f_{p,i}$ is in $\sharpFA$.
    The function $1^n \mapsto 2^n$ is computed by the right NFA in Figure~\ref{fig:thm:univclosure:linear} and thus in $\sharpFA$.
    The remainder of the claim follows by Lemma~\ref{lem:add} and Lemma~\ref{lem:mult}.
    Note that $f_{p,i}(1^n) \equiv_p i$, so $\varphi \circ f_{p,i} = \varphi_i \circ f_{p,i}$ has to be in $\sharpFA$ as well.
    Furthermore $\varphi(f_{p,i}(1^n)) = \varphi_i(f_{p,i}(1^n)) = \varphi_i(p \cdot (2^{n}+ N) +i) \geq 2^{\gamma p \cdot (2^n+N) + \gamma i}$ for all $n \in \IN$ which is larger than any NFA can compute, since NFAs can only compute functions that are at most linearly exponential in the length of the input.
    We conclude that none of the constituents $\varphi_i$ of $\varphi$ can be exponential, so they are instead all polynomials, making $\varphi$ an ultimately PORC function.
\end{proof}

\subsection{Multivariate functional closure properties}

\begin{theorem}[Classification of multivariate functional closure properties of $\sharpFA$]
    \label{thm:multivclosure}
    A function $\varphi: \IN^m \to \IN$ is a functional closure property of $\sharpFA$ iff $\varphi$ can be written as a finite sum of finite products of univariate ultimately PORC functions.
\end{theorem}

\vspace{-0.5cm}

\begin{proof}
    By Theorem~\ref{thm:univclosure} any univariate ultimately PORC function is a closure property of $\sharpFA$.
    As such any finite sum or finite product of them is also a closure property of $\sharpFA$ by Lemmas~\ref{lem:add} and~\ref{lem:mult}.
    
    It remains to show that all functional closure properties of $\sharpFA$ are of this form.
    For this let $\varphi: \IN^m \to \IN$ be a functional closure property of $\sharpFA$.
    Define the alphabet $\Sigma = \{\sigma_1, \ldots, \sigma_m\}$ and the functions $f_i: \Sigma^\star \to \IN$ where $f_i(w) := \#_i(w)$ is defined as the number of occurences $\#_i(w)$ of the symbol $\sigma_i$ in $w$.
    Applying the closure property to $f_1, \ldots, f_m$ gives that $\varphi \circ (f_1, \ldots, f_m) \in \sharpFA$ and thus is computed by an NFA $M = (Q, \Sigma, \wt, \inp, \outp)$.
    This induces transition matrices $A_\sigma \in \IN^{|Q| \times |Q|}$ for each symbol $\sigma \in \Sigma$ and vectors $a, b \in \IN^{|Q|}$, s.t.\ $a^T \prod_{j=1}^{|w|}A_{w_j}b = \varphi(f_1(w), \ldots, f_m(w))$ for all $w \in \Sigma^\star$.
    Restricting to words of the form $w = \sigma_1^{n_1} \sigma_2^{n_2} \cdots \sigma_m^{n_m}$ for $n_1, \ldots, n_m \in \IN$ gives $a^T \left(\prod_{i=1}^{m}A_{\sigma_i}^{n_i}\right) b = \varphi(n_1, \ldots, n_m)$.
    Using Lemma~\ref{lem:powerstructure} on each of the $A_{\sigma_i}^{n_i}$ we see that every entry of $A_{\sigma_i}^{n_i}$ is an ultimately almost PORC function in $n_i$.
    Consequently, every entry of $\prod_{i=1}^{m}A_{\sigma_i}^{n_i}$ is a finite sum of products of different ultimately almost PORC functions and the same holds for $a^T \left(\prod_{i=1}^{m}A_{\sigma_i}^{n_i}\right) b = \varphi(n_1, \ldots, n_m)$.

    We now look at the individual summands of $\varphi$ and prove that we can rewrite each one as a product of ultimately PORC functions by one-by-one rewriting the exponential constituents.
    For this let $\varphi^{(1)}(n_1) \cdots \varphi^{(m)}(n_m)$ be one of the summands of $\varphi$ where $\varphi^{(1)}, \ldots, \varphi^{(m)}$ are all ultimately almost PORC functions, with periods $p_1, \ldots, p_m$, offsets $N_1, \ldots, N_m$ and constituents $\varphi^{(i)}_0, \ldots, \varphi^{(i)}_{p_i-1}$ for each $i \in [m]$.
    If none of the constituents are exponential we are done.
    Otherwise let $\varphi^{(i)}_j$ be one of the exponential constituents, let $\gamma \in \IR^+$ and let $N \in \IN$, s.t. $\varphi^{(i)}_j(n_i) \geq 2^{\gamma n_i}$ for $n_i \geq N$.
    We claim we can set $\varphi^{(i)}_j(n_i) = 0$ without changing the product $\varphi^{(1)}(n_1) \cdots \varphi^{(m)}(n_m)$ for any $n_1, \ldots, n_m \in \IN$.
    Call the resulting functions $\psi^{(i)}$ and $\psi^{(i)}_j$.
    Assume for the sake of contradiction, that there are some $c_1, \ldots, c_m \in \IN$ where
    \mbox{$\varphi^{(1)}(c_1) \cdots \varphi^{(m)}(c_m) \neq \varphi^{(1)}(c_1) \cdots \varphi^{(i-1)}(c_{i-1})\cdot \psi^{(i)}(c_i) \cdot \varphi^{(i+1)}(c_{i+1}) \cdots \varphi^{(m)}(c_m)$}.
    This implies that $\varphi^{(1)}(c_1) \cdots \varphi^{(i-1)}(c_{i-1}) \cdot \varphi^{(i+1)}(c_{i+1}) \cdots \varphi^{(m)}(c_m) \neq 0$ and $c_i \geq N_i$ as we didn't change any other functions except $\varphi^{(i)}$ for $n_i \geq N_i$.

    Constructing constant functions $f'_k(w) = c_k$ for $k \neq i$ and the function $f'_i(w) = p_i \cdot (2^{|x|} + \max(N_i, N)) + j$ which are all in $\sharpFA$. We see that $\varphi \circ (f'_1, \ldots, f'_m) \in \sharpFA$.
    Note that for any $w \in \Sigma^\star$ we have $f'_i(w) \geq \max(N_i, N)$ and $f'_i(w) \equiv_{p_i} j$ and thus $\varphi^{(i)} \circ f'_i = \psi^{(i)}_j \circ f'_i$.
    Combining all of this we again reach a contradiction to the fact that NFAs can only compute at most linearly exponential functions via

\vspace{-0.7cm}

    \begin{align*}
        \varphi(f'_1(w), \ldots, f'_m(w)) &\geq \varphi^{(1)}(c_1) \cdots \varphi^{(i-1)}(c_{i-1}) \cdot \varphi^{(i)}(f'_i(w)) \cdot \varphi^{(i+1)}(c_{i+1}) \cdots \varphi^{(m)}(c_m)\\
                                  &\geq \varphi^{(i)}(f'_i(w))
                                  \ = \ \varphi^{(i)}_j(f'_i(w))
                                \ = \ \varphi^{(i)}_j(p_i\cdot (2^{|w|} + \max(N_i, N)) + j)\\
                                  &\geq 2^{\gamma p_i\cdot (2^{|w|} + \max(N_i, N)) + \gamma j}\,.
    \end{align*}

\vspace{-0.3cm}

    Note that the first inequality holds due to all summands of $\varphi$ being non-negative.
\end{proof}

\vspace{-0.2cm}

Deciding whether a function $\varphi$ has such a representation may not always be directly visible, however if $\varphi$ is a multivariate polynomial we can be more explicit.
Every integer-valued multivariate polynomial has integer coefficients when represented in the binomial basis (see \cite[Prop.\ 4.2.1]{IP22} for a proof of this fact).
We say a term $a \cdot \binom{x_1}{d_1} \cdots \binom{x_m}{d_m}$ dominates another term $a' \cdot \binom{x_1}{d'_1} \cdots \binom{x_m}{d'_m}$ if $d_i \geq d'_i$ for all $i \in [m]$.
A term is a dominating term of $\varphi$ if it has non-zero coefficient and it is not dominated by any other term with non-zero coefficient.
We can use a similar approach to Lemma~\ref{lem:binbasischange} to rewrite $\varphi$ as a positive integer linear combination of products of shifted binomials (details can be found in the appendix).

\vspace{-0.2cm}

\begin{restatable}{lemma}{restatebinbasischangemult}
    \label{lem:binbasischangemult}
    Let $\varphi(x_1, \ldots, x_m) = \sum_{i=1}^r a_i \cdot \prod_{j=1}^m \binom{x_j}{d_{i,j}}$ with $a_i \in \IZ$ and the coefficients of the dominating terms being positive.
    Then there are $a'_1, \ldots, a'_{r'} \in \IN$ and $c_1, \ldots, c_{r'} \in \IN$ with $\varphi(x_1, \ldots, x_m) = \sum_{i=1}^{r'} a'_i \cdot \prod_{j=1}^m \binom{x_j-c_i}{d_{i,j}}$.
\end{restatable}

\vspace{-0.2cm}

We can then use a generalization of Lemma~\ref{lem:polynomial} and replace subtractions $x_i - c'$ by $\max(x_i - c', 0)$.
However special care has to be taken for $x_i < c'$, in which case $x_i$ has to be replaced by the corresponding constants first.
This adds the additional condition on $\varphi$.

\vspace{-0.2cm}

\begin{restatable}{lemma}{restatemultipolynomial}
    \label{lem:multipolynomial}
    Let $\varphi: \IN^m \to \IN$ be a multivariate polynomial with rational coefficients, such that whenever $\psi$ is formed from $\varphi$ by replacing any set of variables -- including the empty subset -- by constants from $\IN$, then all dominating terms of $\psi$ have positive coefficients.
    Then $\varphi$ is a functional closure property of $\sharpFA$.
\end{restatable}

\vspace{-0.5cm}

\begin{restatable}{lemma}{restatemultivarpolyn}
    \label{lem:multivarpolyn}
    A multivariate polynomial $\varphi: \IN^m \to \IN$ with rational coefficients is a functional closure property of $\sharpFA$ iff for every $\psi$ that can be formed from $\varphi$ by replacing any subset of variables -- including the empty set -- by constants from $\IN$, then all dominating terms of $\psi$ have positive coefficients.
\end{restatable}

\vspace{-0.5cm}

\begin{proofsketch}
    Lemma~\ref{lem:multipolynomial} already proves that all multivariate polynomials of this form are functional closure properties of $\sharpFA$.
    Now let $\varphi: \IN^m \to \IN$ be a multivariate polynomial with rational coefficient and a functional closure property of $\sharpFA$.
    By Theorem~\ref{thm:multivclosure} $\varphi$ can be written as a finite sum of finite products of ultimately PORC functions.
    Note that the leading coefficient of each constituent of these ultimately PORC functions is positive.
    By multivariate polynomial interpolation we can now show that $\varphi$ is already a finite sum of finite products of these constituents.
    The dominating terms of $\varphi$ are then formed by products of the leading coefficients of the constituents and thus are positive.
\end{proofsketch}

\section{Promise closure properties}
\label{sec:promises}

\begin{definition}
    Let $S \subseteq \IN^m$ and let $\varphi: \IN^m \to \IN$ be a function.
    We call $\varphi$ a \emph{functional promise closure property} of $\sharpFA$ with regard to $S$ if for every $f_1, \ldots, f_m \in \sharpFA$ defined on some shared alphabet $\Sigma$ there is a function $g \in \sharpFA$ with $g(w) = \varphi(f_1(w), \ldots, f_m(w))$ for every $w \in \Sigma^\star$ for which $(f_1(w), \ldots, f_m(w)) \in S$.
\end{definition}

\vspace{-0.2cm}

We now want to show that if $S$ fulfils some property, namely admitting polynomial cluster sequences (we postpone the definition to Definition~\ref{def:polynomialclustersequences}), then for every functional promise closure property with regard to $S$ there is a functional closure property of $\sharpFA$ that agrees with it on all tuples in $S$.
In other words we can interpolate the functional promise closure property $\varphi$ on all values of $S$ to obtain a functional closure property $\varphi'$ for all of $\sharpFA$.
The proof for this follows along the following ideas:
First, similar to Theorem~\ref{thm:multivclosure}, use the functional closure property $\varphi$ on the unary counting functions to find an equivalent function $\varphi'$ that is almost a functional closure property.
However after this step some of the constituents of the ultimately almost PORC functions might still be exponential.
Theorem~\ref{thm:multivclosure} then proceeded by showing that we can replace all these exponential constituents by the constant zero function, as otherwise we were able to reach a contradiction by constructing functions $f'_1, \ldots, f'_m$ and an infinite sequence of inputs $w^{(i)}$, such that $\varphi(f'_1(w^{(i)}), \ldots, f'_m(w^{(i)}))$ grows doubly exponential in the length of the inputs.
However when dealing with functional promise closure properties we have to be more careful when choosing $f'_1, \ldots, f'_m$ and the $w^{(i)}$, because we need $(f'_1(w^{(i)}), \ldots, f'_m(w^{(i)})) \in S$ to reach a contradiction.
Additionally we don't replace the exponential constituents by the constant zero-function but rather a polynomial that behaves the same for small inputs.
For this we need a special variant of univariate polynomial interpolation that yields integer-valued polynomials that are non-negative for all inputs from $\IN$:

\vspace{-0.2cm}

\begin{restatable}{lemma}{restateinterpolatepositive}
    \label{lem:interpolate_positive}
    Let $c_0, \ldots, c_N \in \IN$.
    Then there is an integer-valued polynomial $q: \IQ \to \IQ$ with $q(n) = c_n$ for all $n \in \{0, \ldots, N\}$ and $q(n') \geq 0$ for all $n' \in \IN$.
\end{restatable}

\vspace{-0.2cm}

To be able to hit all of $S$ consistently we use independent binary encodings, that allow us to hit all of $\IN^m$.

\vspace{-0.2cm}

\begin{restatable}[Folklore]{lemma}{restatebinary}
    \label{lem:binary}
    The function $f: \{0,1\}^\star \to \IN$ defined by being the value of $w \in \{0, 1\}^\star$ interpreted as a binary number is in $\sharpFA$.
    Additionally it is possible to extend the domain of $f$ to any alphabet $\Sigma \supseteq \{0, 1\}$ where the value of $f$ is determined while ignoring any symbols not in $\{0, 1\}$.
\end{restatable}

\vspace{-0.2cm}

For any $n \in \IN$ we denote by $\bin(n)$ the unique binary representation of $n$ without leading zeros.
We first show the methodology in detail by proving the univariate case.

\vspace{-0.2cm}

\begin{theorem}
    Let $S \subseteq \IN$.
    Then any function $\varphi: \IN \to \IN$ is a functional promise closure property of $\sharpFA$ with regard to $S$ iff there is a functional closure property $\psi: \IN \to \IN$ of $\sharpFA$ with $\varphi_{|S} = \psi_{|S}$.
\end{theorem}

\vspace{-0.5cm}

\begin{proof}
    If such a $\psi$ exists, we directly see that $\varphi$ is a functional promise closure property of $\sharpFA$ with respect to $S$. Indeed for any $f \in \sharpFA$ we construct $g = \psi \circ f \in \sharpFA$ and see $g(w) = \varphi(f(w))$ for every $w \in \Sigma^\star$ with $f(w) \in S$.

    Now on the other hand let $\varphi$ be any functional promise closure property of $\sharpFA$ with regard to $S$.
    Again define the alphabet $\Sigma = \{1\}$ and the function $f: \Sigma^\star \to \IN$ with $f(1^n) := n$.
    Applying the closure property to $f$ gives that there is some $g \in \sharpFA$ with $g(1^n) = \varphi(f(1^n)) = \varphi(n)$ for all $n \in S$.
    Let $M = (Q, \Sigma, \wt, \inp, \outp)$ be an NFA computing $g$.
    This induces a transition matrix $A \in \IN^{|Q| \times |Q|}$ and vectors $a, b \in \IN^{|Q|}$, s.t.\ $a^T A^n b = g(1^n)$ for all $n \in \IN$ by Lemma~\ref{lem:algebraic}.
    We now define $\chi(n) := a^T A^n b$ and see $\chi_{|S} = \varphi_{|S}$.
    Using Lemma~\ref{lem:powerstructure} on $A^{n}$ we see that every entry of $A^n$ is an ultimately almost PORC function in $n$ and as such $\psi$ is also an ultimately almost PORC function.
    Let $p$ be the quasiperiod of $\chi$ and let $\chi_0, \ldots, \chi_{p-1}$ be the constituents of $\chi$ and let $N \in \IN$ be the offset after which $\chi$ is defined by the constituents.

    We claim that we can now replace every exponential constituent by a polynomial one without changing the value of $\chi$ for any $n \in S$.
    For each $i \in \{0, \ldots, p-1\}$ we distinguish two cases, depending on whether $S_i := S \cap (p\IZ+i)$ is finite or it is infinite.
    If $S_i$ is a finite set, we replace $\chi'_i$ with a polynomial that interpolates the same values as $\chi_i$ on $S_i$.
    Lemma~\ref{lem:interpolate_positive} ensures that this polynomial is integer-valued and non-negative for all of $\IN$.
    If $S_i$ is an infinite set, we replace $\chi_i$ by the constant zero function.
    Call the resulting ultimately PORC function $\psi$ with constituents $\psi_i$.
    Assume for the sake of contradiction there is an $c \in S$, s.t.\ $\chi(c) \neq \psi(c)$.
    Clearly such a $c$ would have to be at least $N$.
    Now let $i$ be, s.t.\ $c \in S_i$.
    It must hold that $\chi_i(c) \neq \psi_i(c)$.
    Hence $S_i$ cannot be a finite set, since $\chi_i$ and $\psi_i$ agree on $S_i$.
    Therefore $S_i$ must be infinite.
    Since $\chi_i$ is exponential there is a $\gamma \in \IR^+$ and $N' \in \IN$, s.t.\ $\chi_i(n) \geq 2^{\gamma n}$ for all $n \geq N'$.
    Let $f' \in \sharpFA$ be the function of binary evaluation from Lemma~\ref{lem:binary} over the alphabet $\Sigma' = \{0, 1\}$.
    Then there is a $g' \in \sharpFA$ with $g'(\bin(n)) = \varphi(f'(\bin(n))) = \chi(f'(\bin(n)))$ for all $n \in S$.
    Since $S_i$ is infinite, in particular $S_i$ must contain infinitely many values bigger than $\max(N, N')$.
    For $n \in S_i$ with $n \geq \max(N, N')$ we now have

\vspace{-0.6cm}
    
    \begin{align*}
        g'(\bin(n)) = \chi(f'(\bin(n))) = \chi(n) = \chi_i(n) \geq 2^{\gamma n} \geq 2^{\gamma 2^{|\bin(n)|-1}},
    \end{align*}

\vspace{-0.2cm}

    which is a contradiction to NFAs only being able to only compute functions that are at most linearly exponential in the input length.
    In conclusion $S_i$ cannot be infinite either and thus $c$ itself cannot exist.
\end{proof}

\vspace{-0.4cm}

\begin{definition}
    \label{def:polynomialclustersequences}
    A set $S \subseteq \IN^m$ admits \emph{polynomial cluster sequences} if for every $N_1, \ldots, N_m \in \IN$, and $p_1, \ldots, p_m \in \IN$ the projection $\tau: S \to \{0, \ldots, N_1 + p_1 - 1\} \times \ldots \times \{0, \ldots, N_m + p_m - 1\}$ defined by $\tau(n_1, \ldots n_m) = (\smod{n_1}{p_1}{N_1}, \ldots, \smod{n_m}{p_m}{N_m})$ has the following property:
    Any preimage $T$ of a singleton set under $\tau$ for every $i \in [m]$ has either bounded $i$-th coordinate or there is a polynomial $q: \IN \to \IN$ and an infinite subset $T' \subseteq T$ with unbounded $i$-th coordinate\footnote{note that this property also follows directly from the polynomial bound on the other coordinates in the $i$-th coordinate, but we have it as part of the definition for clarity.} and with $\sum_{j=1}^m n_j \leq q(n_i)$ for all $(n_1, \ldots, n_m) \in T'$.
    We call such an infinite subset a \emph{polynomial cluster sequence} with regards to dimension $i$.

    We call $\tau$ the \emph{shifted grid projection} of $S$ with respect to offsets $N_1, \ldots, N_m$ and quasiperiods $p_1, \ldots, p_m$.
\end{definition}

Intuitively this definition requires that each dimension is either bounded or can grow reasonably quickly together with the other dimensions, even when restricted to inputs from specific shifted residue classes.
For example $\{(n^2, n^3) \mid n \in \IN\}$ admits polynomial cluster sequences, while $\{(n, 2^n) \mid n \in \IN\}$ does not.

\vspace{-0.1cm}

\begin{restatable}{theorem}{restatepolyclusterseq}
    \label{thm:polyclusterseq}
    Let $S \subseteq \IN^m$ admit polynomial cluster sequences.
    Then any function $\varphi: \IN^m \to \IN$ is a functional promise closure property of $\sharpFA$ with regard to $S$ iff there is a functional closure property $\psi: \IN^m \to \IN$ of $\sharpFA$ with $\varphi_{|S} = \psi_{|S}$.
\end{restatable}

\vspace{-0.1cm}

The technical proof of the theorem can be found in the appendix.
There are multiple natural families for the set $S$ such that $S$ admits polynomial cluster sequences which we describe in the following.

\vspace{-0.1cm}

\begin{lemma}
    Any finite set $S \subseteq \IN^m$ admits polynomial cluster sequences.
\end{lemma}

\vspace{-0.4cm}

\begin{proof}
    Independent of the offsets and quasiperiods and $i \in [m]$ any subset of $S$ always is finite and thus bounded in every dimension.
\end{proof}

\vspace{-0.2cm}

An affine variety is defined as the zero set of a finite number of multivariate polynomials.
A special case of affine varieties are \emph{graph varieties} (also just  called \emph{graphs}, see \cite[\S2.4, Exe.~12]{Sha13}).
An affine variety $S$ is a graph variety if there exist a finite number of $j$-variate polynomials $\mu_1,\ldots,\mu_k$ such that $S = \{(s_1,\ldots,s_j,\mu_1(s_1,\ldots,s_j),\ldots,\mu_k(s_1,\ldots,s_j))\mid (s_1,\ldots,s_j)\in\IQ^j\} \subseteq \IQ^{j+k}$.
We call $s_1, \ldots, s_j$ the \emph{free} variables, and the remaining variables the \emph{dependent} variables.
We call $S$ a monotone graph variety if $\mu_1, \ldots, \mu_k$ are all monotone.

\vspace{-0.1cm}

\begin{restatable}{lemma}{restatemonotonegraph}
    Let $S \subseteq \IQ^m=\IQ^{j+k}$ be a monotone graph variety.
    Then the set $S \cap \IN^m=\IN^{j+k}$ admits polynomial cluster sequences.
\end{restatable}

\vspace{-0.2cm}

The proof of this lemma can be found in the appendix.
All in all this combines to the following theorem which characterizes the special case of multivariate polynomial functional promise closure properties, the technical details can be found in the appendix.

\vspace{-0.1cm}

\begin{restatable}{theorem}{restatemongraphclosures}
    \label{thm:mongraphclosures}
    Let $S \subseteq \IQ^m=\IQ^{j+k}$ be a monotone graph variety and let $I = I(S)$ be its vanishing ideal.
    A multivariate polynomial $\varphi: S \to \IN$ is a functional promise closure property of $\sharpFA$ with regard to $S$ if and only if there exists a $\psi \in I$ such that $\varphi + \psi$ is a multivariate functional closure property of $\sharpFA$.
\end{restatable}

\section{Conclusion}
We characterized the functional closure properties of $\sharpFA$ to be precisely the ultimately PORC functions in the univariate case and combinations of ultimately PORC functions in the multivariate case.
Additionally we characterize promise functional closure properties of $\sharpFA$ with regard to some natural families of sets $S$.
Natural further directions of research are now whether we can characterize the promise functional closure properties of $\sharpFA$ for more sets $S$ and whether our methods can be applied to characterize functional closure properties for more powerful models of computation.

\bibliography{arxiv}

\newpage

\appendix
\section{Appendix: Missing proofs}
\subsection{Missing proofs from section~\ref{sec:notation}}
\begin{proof}[Proof of Lemma~\ref{lem:simplify}]
    Let $M_f = (Q_f, \Sigma, \wt_f, \inp_f, \outp_f)$ be an NFA computing $f$.
    We construct $M = (Q, \Sigma, \wt, \inp, \outp)$ with
    \begin{align*}
        Q &= Q_f \times [\max (\image \wt_f \cup \image \inp_f)] \times [\max(\image \outp_f)]\\
        \wt((q, \alpha, \beta), \sigma, (q', \alpha', \beta')) &=
        \begin{cases}
            1 & \text{if $\alpha' \leq \wt_f(q, \sigma, q')$ and $\beta=1$}\\
            0 & \text{otherwise}\\
        \end{cases}\\
        \inp((q, \alpha, \beta)) &=
        \begin{cases}
            1 & \text{if $\alpha \leq \inp_f(q)$}\\
            0 & \text{otherwise}\\
        \end{cases}\\
        \outp((q, \alpha, \beta)) &=
        \begin{cases}
            1 & \text{if $\beta \leq \outp_f(q)$}\\
            0 & \text{otherwise}\\
        \end{cases}\\
    \end{align*}
    Intuitively $M$ consists of multiple identical copies of $M_f$, where all the weights are handled explicitly as separate computations.

    Fix any word $w = w_1 \ldots w_n \in \Sigma^\star$.
    Let $P = q_0q_1 \ldots q_n$ be a computation of $M_f$ on $w$.
    Then for $\alpha_0 \in [\inp_f(q_0)]$ and $\alpha_i \in [\wt(q_{i-1},w_i,q_i)]$ for $i \in [n]$ and $\beta_0 = \beta_1 = \ldots = \beta_{n-1} = 1$ and $\beta_n \in \outp_f(q_n)$ we have that $P' := (q_0,\alpha_0,\beta_0)(q_1,\alpha_1,\beta_1) \ldots (q_n,\alpha_n,\beta_n)$ is a computation of $M$ on $w$ of weight exactly $1$ while any other choices of the $\alpha_i$ and $\beta_i$ would lead to a computation of weight $0$.
    In particular there are $\inp_f(q_0) \cdot \prod_{i=1}^{n}\wt_f(q_{i-1},w_i,q_i) \cdot \outp_f(q_0)$ non-zero computations in $M$ corresponding to $P$, which coincides with the weight of $P$ in $M_f$.
    Additionally every computation on $M$ can be projected onto a computation on $M_f$ via $(q, i, j) \mapsto q$ and thus $M$ computes $f$.
\end{proof}

\subsection{Missing proofs from section~\ref{sec:functional_closure_properties}}
\restateadd*
\begin{proof}
    Let $M_f = (Q_f, \Sigma, \wt_f, \inp_f, \outp_f)$ and $M_g = (Q_g, \Sigma, \wt_g, \inp_g, \outp_g)$ be NFAs computing $f$ and $g$ respectively.

    We construct the direct union NFA $M = (Q, \Sigma, \wt, \inp, \outp)$ with
    \begin{align*}
        Q &= Q_f \dot\cup Q_g\\
        \wt(q, \sigma, q') &=
        \begin{cases}
            \wt_f(q, \sigma, q') & \text{$q, q' \in Q_f$}\\
            \wt_g(q, \sigma, q') & \text{$q, q' \in Q_g$}\\
            0 & \text{otherwise}\\
        \end{cases}\\
        \inp(q) &=
        \begin{cases}
            \inp_f(q) & \text{if $q \in Q_f$}\\
            \inp_g(q) & \text{if $q \in Q_g$}\\
        \end{cases}\\
        \outp(q) &=
        \begin{cases}
            \outp_f(q) & \text{if $q \in Q_f$}\\
            \outp_g(q) & \text{if $q \in Q_g$}\\
        \end{cases}
    \end{align*}

    Any computation with non-zero weight in $M$ is now fully contained in exactly one of the sub-NFA equivalent to $M_f$ or $M_g$ and any computation path on $M_f$ or $M_g$ naturally embeds into $M$ and thus $M$ computes $f+g$.
\end{proof}

\restatemult*
\begin{proof}
    Let $M_f = (Q_f, \Sigma, \wt_f, \inp_f, \outp_f)$ and $\wt_g = (Q_g, \Sigma, \wt_g, \inp_g, \outp_g)$ be NFAs computing $f$ and $g$ respectively.
    We use the product automaton construction to construct the NFA $M = (Q, \Sigma, \wt, \inp, \outp)$:
    \begin{align*}
        Q &= Q_f \times Q_g\\
        \wt((q_1, q_2), \sigma, (q_1', q_2')) &= \wt_f(q_1, \sigma, q_1') \cdot \wt_g(q_2, \sigma, q_2')\\
        \inp((q_1, q_2)) &= \inp_f(q_1) \cdot \inp_g(q_2)\\
        \outp((q_1, q_2)) &= \outp_f(q_1) \cdot \outp_g(q_2)\\
    \end{align*}

    Fix some $w = w_1 \ldots w_n \in \Sigma^\star$.
    Let $\mathcal{P}_f = Q_f^{n+1}$ be all computations of $M_f$ on $w$ and let $\mathcal{P}_g = Q_g^{n+1}$ be all computations of $M_g$ on $w$.
    Then $\mathcal{P}_f \times \mathcal{P}_g \cong (Q_f \times Q_g)^{n+1}$ are precisely the computations of $M$ on $x$ by identifying $((q_{f,0} \ldots q_{f,n}), (q_{g,0} \ldots q_{g,n}))$ with the computation $(q_{f,0}, q_{g,0}) \ldots (q_{f,n}, q_{g,n})$.
    Additionally the computation $(P_f, P_g) \in \mathcal{P}_f \times \mathcal{P}_g$ has weight $\Cweight((P_f, P_g)) = \Cweight_f(P_f) \cdot \Cweight_g(P_g)$ where $\Cweight_f$ and $\Cweight_g$ denote the weights of computations of $M_f$ and $M_g$ respectively.
    Thus we have
    \[
        \Bigl(\sum_{P_f \in \mathcal{P}_f}\Cweight_f(P_f)\Bigr) \cdot \Bigl(\sum_{P_g \in \mathcal{P}_g}\Cweight_g(P_g)\Bigr) = \sum_{P_f \in \mathcal{P}_f}\sum_{P_g \in \mathcal{P}_g}\Cweight_f(P_f)\cdot\Cweight_g(P_g) = \sum_{(P_f, P_g) \in \mathcal{P}_f \times \mathcal{P}_g}\Cweight_f((P_f, P_g))
    \]
    and $M$ computes $f \cdot g$.
\end{proof}

\restatebaseconstruction*
\begin{proof}
    Let $M_f = (Q_f, \Sigma, \wt_f, \inp_f, \outp_f)$.
    We construct the NFA $M = (Q, \Sigma, \wt, \inp, \outp)$ with $Q = Q_f \times S$ and
    \begin{align*}
        \inp((q, s)) &= \begin{cases}
            \inp_f(q) & \text{if $s = \init(q)$}\\
            0 & \text{otherwise}\\
        \end{cases}\\
        \outp((q, s)) &= \outp_f(q) \cdot \cond(s)\\
        \wt((q, s), \sigma, (q', s')) &= \begin{cases}
            \wt_f(q, \sigma, q') & \text{if $s' = \step(q, \sigma, q', s)$}\\
            0 & \text{otherwise}\\
        \end{cases}
    \end{align*}
    Fix some $w = w_1 \ldots w_n \in \Sigma^\star$.
    Any computation $P = (q_0, s_0) \ldots (q_n, s_n)$ of $M$ on $w$ directly corresponds to the computation $q_0 \ldots q_n$ of $M_f$.
    If $s_0 \neq \init(q_0)$ or if $s_i \neq \step(q_{i-1}, w_i, q_i, s_{i-1})$ for any $i \in [n]$ then $P$ has weight $0$.
    Otherwise $P$ has weight $\Cweight(P) = \Cweight_f(q_0 \ldots q_n) \cdot \cond(s_n) = \Cweight_f(q_0 \ldots q_n) \cdot \prop(w, q_0 \ldots q_n)$.
    The claim on the function computed by $M$ directly follows.
\end{proof}

\restatebaseconstructionall*
\begin{proof}
    Let $M_f = (Q_f, \Sigma, \wt_f, \inp_f, \outp_f)$.
    We construct the NFA $M = (Q, \Sigma, \wt, \inp, \outp)$ with $Q = Q_f \times S$ and
    \begin{align*}
        \inp((q, s)) &= \begin{cases}
            1 & \text{if $s = \init(q)$}\\
            0 & \text{otherwise}\\
        \end{cases}\\
        \outp((q, s)) &= \cond(s)\\
        \wt((q, s), \sigma, (q', s')) &= \begin{cases}
            1 & \text{if $s' = \step(q, \sigma, q', s)$}\\
            0 & \text{otherwise}\\
        \end{cases}\\
    \end{align*}
    Any computation $P = (q_0, s_0) \ldots (q_n, s_n)$ of $M$ directly corresponds to the computation $q_0 \ldots q_n$ of $M_f$.
    If $s_0 \neq \init(q_0)$ or if $s_i \neq \step(q_{i-1}, w_i, q_i, s_{i-1})$ for any $i \in [n]$ then $P$ has weight $0$.
    Otherwise $P$ has weight $\Cweight(P) = \cond(s_n) = \prop(w, q_0 \ldots q_n)$.
    The claim on the function computed by $M$ directly follows.
\end{proof}

\restatepropall*
\begin{proof}
    We construct $\prop = (S, \init, \step, \cond)$ as follows:

\vspace{-0.7cm}

    \begin{align*}
        S &= Q \to \CR\\
        \init(q) &= r \mapsto \tau(\inp(r))\\
        \step(q, \sigma, q', s) &= r \mapsto \textstyle\sum_{r' \in Q}s(r') \cdot \tau(\wt(r', \sigma, r))\\
        \cond(s) &= \pi\Bigl(\textstyle\sum_{r \in Q}s(r) \cdot \tau(\outp(r))\Bigr)
    \end{align*}

\vspace{-0.3cm}

    Note that $\step$ and $\init$ depend on neither $q$ nor $q'$.
    So for a fixed $w = w_1 \ldots w_n \in \Sigma^\star$ the step sequence $s_0, \ldots, s_n$ collapses to $s_0 = r \mapsto \tau(\inp(r))$ and $s_i = r \mapsto \sum_{r' \in Q}s_{i-1}(r') \cdot \tau(\wt(r', w_i, r))$ for $i \in [n]$ independent of the specific computation of $N$ on $x$.

    We first prove that for all $w = w_1\ldots w_n \in \Sigma^\star$ and $q \in Q$ we have $s_n(r) = \tau(\sum_{P_r} \Cpweight(P_r))$ where the sum is over all computations $P_r=q_0 \ldots q_n$ of $N$ on $w$ with $q_n = r$.
    We prove this by induction over $n$.
    For $n = 0$ this holds directly.
    For $n>0$ we have
    \begin{align*}
        s_n(r) &= \sum_{q_{n-1} \in Q}s_{n-1}(q_{n-1}) \cdot \tau(\wt(q_{n-1}, w_n, r))\\
               &= \sum_{q_{n-1} \in Q}\tau\Bigl(\sum_{q_0, \ldots, q_{n-2} \in Q}\Cpweight(q_0\ldots q_{n-1})\Bigr) \cdot \tau(\wt(q_{n-1}, w_n, r))\\
               &= \tau\Bigl(\sum_{q_{n-1} \in Q}\sum_{q_0, \ldots, q_{n-2} \in Q}\Cpweight(q_0\ldots q_{n-1}) \cdot \wt(q_{n-1}, w_n, r)\Bigr)\\
               &= \tau\Bigl(\sum_{q_0, \ldots, q_{n-1} \in Q}\Cpweight(q_0\ldots q_{n-1}r)\Bigr)\,.\\
    \end{align*}

\vspace{-0.8cm}

    We can thus finish the proof with

\vspace{-0.7cm}

    \begin{align*}
        \prop(w, P) &= \cond(s_n)\\
                    &= \pi\Bigl(\sum_{q_n \in Q}s_n(r) \cdot \tau(\outp(q_n))\Bigr)\\
                    &= \pi\Bigl(\sum_{q_n \in Q}\tau\Bigl(\sum_{q_0, \ldots, q_{n-1} \in Q}\Cpweight(q_0\ldots q_n)\Bigr) \cdot \tau(\outp(q_n))\Bigr)\\
                    &= \pi\Bigl(\tau\Bigl(\sum_{q_n \in Q}\sum_{q_0, \ldots, q_{n-1} \in Q}\Cpweight(q_0\ldots q_n) \cdot \outp(q_n)\Bigr)\Bigr)\\
                    &= \pi\Bigl(\tau\Bigl(\sum_{q_0, \ldots, q_n \in Q}\Cweight(q_0\ldots q_n)\Bigr)\Bigr)\,.\qedhere
    \end{align*}
\end{proof}

\restateproplexi*
\begin{proof}
    We construct $\prop = (S, \init, \step, \cond)$ as follows:

\vspace{-0.7cm}
    
    \begin{align*}
        S &= Q \to \CR\\
        \init(q) &= r \mapsto \tau(\ONE_{r < q} \cdot \inp(r))\\
        \step(q, \sigma, q', s) &= r \mapsto \textstyle\sum_{r' \in Q}s(r') \cdot \tau(\wt(r', \sigma, r)) + \tau(\ONE_{r < q'} \cdot \wt(q, \sigma, r))\\
        \cond(s) &= \pi(\textstyle\sum_{r \in Q}\tau(\outp(r)) \cdot s(r))
    \end{align*}

    Fix some $w = w_1 \ldots w_n \in \Sigma^\star$ and some computation $P = q_0 \ldots q_n$ of $N$ on $w$ of non-zero weight.
    This fixes the step sequence $s_0, \ldots, s_n$.
    We now prove by induction over $i$ that $s_{i}(r) = \sum_{P'} \tau(\Cpweight(P'))$ for all $i \in \{0, \ldots, n\}$ and $r \in Q$, where the sum is over all computations $P' = q_0' \ldots q_{i}'$ of $w_1 \ldots w_{i}$ with $q_{i}' = r$ that are lexicographically smaller than $q_0 \ldots q_{i}$.
    For $i = 0$ this is trivially the case since the only computations $P' = q_0'$ that are lexicographically smaller than the computation $q_0$ are the ones with $q_0' < q_0$.
    For $i > 0$ for a computation $P' = q_0' \ldots q_{i}'$ to be lexicographically smaller than $q_0 \ldots q_{i}$ there are two possibilities.
    Either $q_0' \ldots q_{i-1}'$ is already lexicographically smaller than $q_0 \ldots q_{i-1}$  or $q_0' \ldots q_{i-1}' = q_0 \ldots q_{i-1}$ and $q_{i}' < q_{i}$.
    In the second case the weight of $P'$ is precisely $\wt(q_{i-1}, w_{i}, q_{i}')$ since $P$ is a computation of non-zero weight and thus weight exactly $1$.
    Combining all of this we can finish the proof of the claim with a similar argument to Lemma~\ref{lem:prop_all}:
    \begin{align*}
        s_{i}(r) &= \sum_{q_{i-1}' \in Q}s_{i-1}(q_{i-1}') \cdot \tau(\wt(q_{i-1}', w_{i}, r)) + \tau(\ONE_{r < q_{i}} \cdot \wt(q_{i-1}, \sigma, r))\\
               &= \tau\Bigl(\sum_{\substack{q_0' \ldots q_{i-1}'\\ \text{lex. smaller than} \\ q_0 \ldots q_{i-1}}}\Cpweight(q_0'\ldots q_{i-1}')) \cdot \wt(q_{i-1}', w_{i}, r)\Bigr) + \tau(\ONE_{r < q_{i}} \cdot \Cpweight(q_0\ldots q_{i-1}r))\\
               &= \tau\Bigl(\sum_{\substack{q_0' \ldots q_{i-1}'\\ \text{lex. smaller than} \\ q_0 \ldots q_{i-1}}} \Cpweight(q_0'\ldots q_{i-1}'r) + \ONE_{r < q_{i}} \cdot \Cpweight(q_0\ldots q_{i-1}r))\Bigr)\\
               &= \tau\Bigl(\sum_{\substack{q_0' \ldots q_{i-1}'r \\ \text{lex. smaller than} \\ q_0 \ldots q_{i-1}q_i }}\Cpweight(q_0'\ldots q_{i-1}'r)\Bigr)
    \end{align*}
    The statement of this lemma then follows analogously to Lemma~\ref{lem:prop_all}.
\end{proof}

\restateclamping*
\begin{proof}
    Let $M_f = (Q_f, \Sigma, \wt_f, \inp_f, \outp_f)$ be a simple NFA computing $f$ with an arbitrary ordering $<$ on $Q_f$.
    Lemma~\ref{lem:prop_lexi} on the capped semiring $\CR_{c}$ with $\pi(a) = \ONE_{a < c}$ for all $a \in \CR$ constructs a stepwise computation property $\prop$ with
    \[
        \prop(w, P) = \begin{cases}
            1 & \begin{minipage}{11cm}
                if the number of non-zero computations $P'$ on $w$ that are lex. smaller\\[-3pt] than $P$ is at most $c-1$
            \end{minipage}\\[3pt]
            0 & \text{otherwise}\\
        \end{cases}
    \]
    for all computations $P$ of $M_f$ on $w$ of non-zero weight, i.e.\ $\prop(w, P) = 1$ iff $P$ is one of the $c$ lexicographically smallest computations on $w$ with non-zero weight.
    It follows that the NFA $M$ constructed by Lemma~\ref{lem:baseconstruction} computes $g(w) = \min(f(w), c)$.
\end{proof}

\restateconst*
\begin{proof}
    Let $M_f = (Q_f, \Sigma, \wt_f, \inp_f, \outp_f)$ be a simple NFA computing $f$ with an arbitrary ordering $<$ on $Q_f$.
    Lemma~\ref{lem:prop_all} on the capped semiring $\CR_{c+1}$ with $\pi_{=}(a) = \ONE_{a = c}$, $\pi_{\leq}(a) = \ONE_{a \leq c}$ and $\pi_{\geq}(a) = \ONE_{a \geq c}$ for all $a \in \CR$ constructs stepwise computation properties $\prop_=$, $\prop_{\leq}$ and $\prop_{\geq}$ respectively with
    \begin{align*}
        \prop_=(w, P) &= \begin{cases}
            1 & \text{if the number of non-zero computations $P'$ on $w$ is exactly $c$}\\
            0 & \text{otherwise}\\
        \end{cases}\\
        \prop_\leq(w, P) &= \begin{cases}
            1 & \text{if the number of non-zero computations $P'$ on $w$ is at most $c$}\\
            0 & \text{otherwise}\\
        \end{cases}\\
        \prop_\geq(w, P) &= \begin{cases}
            1 & \text{if the number of non-zero computations $P'$ on $w$ is at least $c$}\\
            0 & \text{otherwise}\\
        \end{cases}
    \end{align*}
    for all computations $P$ of $M_f$ on $w$.
    It follows that the NFAs $M_=, M_\leq$ and $M_\geq$ constructed by Lemma~\ref{lem:baseconstruction_all} compute $g_=(w) = \ONE_{f(w)=c} \cdot |Q_f|^{|w|+1}$, $g_\leq(w) = \ONE_{f(w)\leq c} \cdot |Q_f|^{|w|+1}$ and $g_\geq(w) = \ONE_{f(w)\geq c} \cdot |Q_f|^{|w|+1}$ respectively.
    Since $|Q_f|^{|w|+1}$ is always positive Lemma~\ref{lem:clamping} finishes the proof with
    \begin{align*}
        \min(g_=(w), 1) &= \ONE_{f(w)=c} &
        \min(g_\leq(w), 1) &= \ONE_{f(w)\leq c} &
        \min(g_\geq(w), 1) &= \ONE_{f(w)\geq c}\,.\qedhere
    \end{align*}
\end{proof}

\restatedivision*
\begin{proof}
    For $c=1$ the statement is trivially true, so let $c \geq 2$.
    Let $M_f = (Q_f, \Sigma, \wt_f, \inp_f, \outp_f)$ be a simple NFA computing $f$ with an arbitrary ordering $<$ on $Q_f$.
    We use the finite cyclic semiring $\CR = \IZ_c$ with the usual addition and multiplication.
    Lemma~\ref{lem:prop_lexi} with $\pi(a) = \ONE_{a = c-1}$ for all $a \in \CR$ constructs a stepwise computation property $\prop$ with
    \[
        \prop(w, P) = \begin{cases}
            1 & \text{if the number of non-zero computations $P'$ on $w$ that are}\\
              &\text{lex. smaller than $P$ is one less than a multiple of $c$}\\
            0 & \text{otherwise}\\
        \end{cases}
    \] for all computations $P$ of $M_f$ on $w$ of non-zero weight.
    It follows that the NFA $M$ constructed by Lemma~\ref{lem:baseconstruction} computes $g(w) = \left\lfloor \frac{f(w)}{c} \right\rfloor$.
\end{proof}

\restatemod*
\begin{proof}
    For $c=1$ the statement is trivially true, so let $c \geq 2$.
    Let $M_f = (Q_f, \Sigma, \wt_f, \inp_f, \outp_f)$ be a simple NFA computing $f$ with an arbitrary ordering $<$ on $Q_f$.
    We use the finite cyclic semiring $\CR = \IZ_c$ with the usual addition and multiplication.
    Lemma~\ref{lem:prop_all} with $\pi(a) = \ONE_{a = d}$ for all $a \in \CR$ constructs a stepwise computation property $\prop$ with
    \[
        \prop(w, P) = \begin{cases}
            1 & \text{if the number of non-zero computations $P'$ on $w$}\\
              &\text{is $d$ more than a multiple of $c$}\\
            0 & \text{otherwise}\\
        \end{cases}
    \] for all computations $P$ of $M_f$ on $w$ of non-zero weight.
    It follows that the NFA $M$ constructed by Lemma~\ref{lem:baseconstruction_all} computes $g(w) = \ONE_{f(w) \equiv_c d} \cdot |Q_f|^{|w|+1}$.
    Since $|Q_f|^{|w|+1}$ is always positive Lemma~\ref{lem:clamping} finishes the proof with
$
        \min(g(w), 1) = \ONE_{f(w) \equiv_c d}
        $.
\end{proof}

\restatebinbasischange*
\begin{proof}
    We prove this via induction on $r$.
    For $r=0$ we can choose $b_0 = a_0$ and $c_0 = 0$ to see that $\varphi$ is of the required form.
    For $r>0$ we choose $b_r = a_r$ and $c_r = \max\{-a_{r-1} + 1, 0\}$.

    We define the polynomial $\varphi'(x) := \varphi(x) - b_r \cdot \binom{x-c_r}{r}$ and analyze the leading binomial $\binom{x-c_r}{r}$.
    Using the Chu-Vandermonde identity \cite{Spi16} we see $\binom{x-c_r}{r} = \sum_{i=0}^r\binom{-c_r}{r-i} \binom{x}{i} = \sum_{i=0}^r(-1)^{r-i} \binom{r-i+c_r-1}{r-i} \binom{x}{i}$.
    
    Combining this with \mbox{$\varphi(x) = \sum_{i=0}^r a_i \cdot \binom{x}{i}$} we see that \mbox{$\varphi(x) - b_r \cdot \binom{x-c_r}{r} = \sum_{i=0}^{r-1} a'_i \cdot \binom{x}{i}$} with $a'_i = a_i - b_r \cdot (-1)^{r-i} \cdot \binom{r-i+c_r-1}{r-i}$.
    Furthermore we have $a'_{r-1} > 0$ due to $b_r \cdot (-1)^{r-(r-1)}\binom{r-(r-1)+c_r-1}{r-(r-1)} = -b_rc_r < a_{r-1}$.
    By induction there are $b'_0, \ldots, b'_{r-1} \in \IN$ and $c'_0, \ldots, c'_{r-1} \in \IN$ s.t.\ $\varphi'(x) = \sum_{i=0}^{r-1}b'_i \cdot \binom{x-c'_i}{i}$.
    Thus by choosing $c_i = c'_i$ and $b_i = b'_i$ for $i \in \{0, \ldots, r-1\}$ we get $\varphi(x) = \sum_{i=0}^r b_i \cdot \binom{x-c_i}{i}$.
\end{proof}

\restatediagonalstructure*
\begin{proof}
    To prove the lemma, let $v$ be arbitrary and let $f_v(n) = (A^n)_{vv}$.
    Note that we always have $f_v(0) = 1$ since $A^0$ is the identity matrix.
    We split into three cases:
    \begin{description}
        \item[$f_v(n) = 0$ for all $n > 0$:] In this case we are done.
        \item[$f_v(n) \in \{0, 1\}$ for all $n \in \IN$ and there is some $\ell > 0$ with $f_v(\ell) = 1$:]
            Let $p \geq 1$ be minimal with $f_v(p) = 1$.
            Then we directly know that for any multiple $i \cdot p$ we have $f_v(ip) = (A^{ip})_{vv} \geq ((A^{p})_{vv})^i = f_v(p)^i = 1$.
            In the setting if $A$ is interpreted as the adjacency matrix of a multigraph, this corresponds to a walk consisting of repeating the same path of length $p$ from $v$ to itself $i$ times.

            Assume for the sake of contradiction there is any other $n$ that is not a multiple of $p$ with $f_v(n) = 1$.
            Then we have that $f_v(n \cdot p) = (A^{np})_{vv} \geq (A^{p})_{vv}^n + (A^{n})_{vv}^p = f_v(p)^n + f_v(n)^p = 2$, a contradiction.
            Going back to the graph setting, this corresponds to the choice of either repeating a walk from $v$ to $v$ of length $n$ a total of $p$ times or repeating a walk from $v$ to $v$ of length $p$ a total of $n$ times, leading to two different walks of length $n \cdot p$ from $v$ to $v$ since neither does $p$ divide $n$ nor vice-versa.
        \item[There is some $\ell \in \IN$ with $f_v(\ell) \geq 2$:]
            Let $p$ be the gcd of the set $Z = \{n \in \IN \mid f_v(n) \geq 1\}$.
            If $n \in \IN$ is not divisible by $p$ we directly have $f_v(n) = 0$ since this implies $n \notin Z$.

            Additionally there is some finite subset $Z' \subseteq Z$ with the same gcd $p$.
            By Bézout's lemma there is some $N \in \IN$ which is a multiple of $p$ s.t.\ that every number of the form $N + i \cdot p$ with $i \in \IN$ can be written as a non-negative linear combination of elements from $Z'$.
            Let $n = N + i \cdot p$ for any $i \in \IN$.
            We can write
            \begin{align}
                n &= N + \underbrace{\frac{n - N - \lfloor \frac{n - N}{\ell} \rfloor \cdot \ell}{p}}_{\alpha :=} \cdot p + \underbrace{\left\lfloor \frac{n - N}{\ell} \right\rfloor}_{\beta :=} \cdot \ell\,.
            \end{align}
            since $\ell \in Z$ implies that $p$ divides $\ell$.
            Furthermore both $\alpha$ and $\beta$ are non-negative integers and $\beta$ is maximal with this property.
            In particular we can write $N + \alpha \cdot p$ as a non-negative linear combination with elements from $Z'$.
            This gives $N + \alpha \cdot p = \sum_{p_j \in Z'} \kappa_j \cdot p_j$ for $\kappa_j \in \IN$.
            Putting this together we have
            \begin{align*}
                f_v(n) &= (A^n)_{vv}\\
                       &= (A^{N + \alpha \cdot p + \beta \cdot \ell})_{vv}\\
                       &= (A^{\sum_{p_j \in Z'} \kappa_j \cdot p_j + \beta \cdot \ell})_{vv}\\
                       &\geq (A^{\beta \cdot \ell})_{vv} \cdot \prod_{p_j \in Z'}(A^{\kappa_j \cdot p_j})_{vv}\\
                       &\geq (A^{\ell})_{vv}^\beta \cdot \prod_{p_j \in Z'}(A^{p_j})_{vv}^{\kappa_j}\\
                       &= f_v(\ell)^\beta \cdot \prod_{p_j \in Z'}f_v(p_j)^{\kappa_j}\\
                       &\geq f_v(\ell)^\beta\\
                       &\geq 2^\beta\\
                       &\geq 2^{\frac{n - N}{\ell}-1} \in 2^{\Omega(n)}\\
            \end{align*}
            Note that no entry in $A^n$ can grow faster than linearly exponential, so this lower bound directly induces a function $g_v \in 2^{\Theta(n)}$ rather than a function $g_v \in 2^{\Omega(n)}$.

            In the graph setting this corresponds to repeating different walks from $v$ to $v$ of length $\ell$ for most of the walk, thus accumulating a linearly exponential amount of different walks and then finishing with some bounded amount of walks from $v$ to $v$ with lengths in $Z'$.
    \end{description}
\end{proof}

\restatepowerstructure*
\begin{proof}
    We use the description of the diagonal entries of the matrices $A_S^n$ given by Lemma~\ref{lem:diagonalstructure} for all sets $S \subseteq [k]$ where $A_S$ is the matrix $A$ obtained by removing the rows and columns given by $S$.
    These correspond to the number of walks of length $n$ from a vertex $v$ to itself in the multigraph associated with the adjacency matrix $A$, restricted to the induced subgraph spanned by all the vertices not in $S$.
    For any $S \subseteq [k]$ we will index $A_S$ the same way as we index $A$, i.e.\ $(A_S)_{ij} = A_{ij}$ whenever $i,j \in [k] \setminus S$.
    Lemma~\ref{lem:diagonalstructure} in particular already proves this lemma for the diagonal entries.

    Let $G = (V, E)$ with $V = [k]$ be the multigraph defined by the adjacency matrix $A$.
    Let $W$ be any walk on $G$.
    Then we can obtain a unique simple path $P_W$ from $W$ by repeatedly contracting cycles in $W$.
    We call two walks $W_1$ and $W_2$ equivalent if $P_{W_1} = P_{W_2}$.
    Note that this forms an equivalence relation on the walks on $G$ and there is a bijection between paths on $G$ and equivalence classes of walks on $G$ and we can construct every walk in the equivalence class $\mathcal{W}_P$ corresponding to a path\footnote{A path here is a sequence of edges, so in particular two paths with the same sequence of vertices do not have to be the same since $G$ is a multigraph} $P = v_1 \to \ldots \to v_r$ as follows:
    At every vertex $v_i$ we can add any walk from $v_i$ back to $v_i$ of any length $\ell_i$, given that it doesn't pass through any of the vertices $v_1, \ldots, v_{i-1}$.
    Note that this restriction is not necessary in general, but it ensures that every walk is uniquely constructed, a necessity for our proof.
    This construction gives a walk of length $r - 1 + \sum_{i = 1}^{r}\ell_i$.

    Fix some $u \neq v \in V$ and let $P = v_1 \to \ldots \to v_r$ be any path from $u$ to $v$, i.e.\ $v_1 = u$ and $v_r = v$.
    Then the class $\mathcal{W}_P$ contains exactly
    \[
        f_{P}(n) := \sum_{(\ell_1, \ldots, \ell_r)} \prod_{i=1}^{r} (A_{S_i}^{\ell_i})_{v_iv_i}
    \]
    walks of length $n \in \IN$ where the sum is over all $(\ell_1, \ldots, \ell_r) \in \IN^r$ with $n = r - 1 + \sum_{i = 1}^{r}\ell_i$ and we define $S_i = \{v_1, \ldots, v_{i-1}\}$.
    By Lemma~\ref{lem:diagonalstructure} we know the structure of each of the $(A_{S_i}^{\ell_i})_{v_iv_i}$ as a function in $\ell_i$.
    Call the $v_i$ either $0$-vertices, $1$-vertices or $\exp$-vertices respectively for each of the three cases of Lemma~\ref{lem:diagonalstructure}.
    Additionally denote by $p_i$ the corresponding period for the $1$-vertices and $\exp$-vertices and in case of $\exp$-vertices denote by $N_i$ the smallest multiple of $p_i$, s.t.\ $(A_{S_i}^{N_i + j \cdot p_i})_{v_iv_i} > 0$ for all $j \in \IN$.
    Define $I_0 := \{i \in [r] \mid v_i \text{\ is a $0$-vertex}\}$ and let $I_1$ and $I_{\exp}$ be defined analogously.

    We first look at the case $I_{\exp} = \emptyset$.
    Then $\prod_{i=1}^r(A_{S_i}^{\ell_i})_{v_iv_i} = 1$ iff $\ell_i = 0$ for all $i \in I_0$ and $\ell_i$ is a multiple of $p_i$ for all $i \in I_1$, otherwise $\prod_{i=1}^r(A_{S_i}^{\ell_i})_{v_iv_i} = 0$.
    Thus $f_P(n)$ is given as the number of integer points in the polytope defined by $\sum_{i \in I_1} \ell_i \cdot p_i = n - r + 1$ and $\ell_i \geq 0$ for all $i \in I_1$.
    This is a scaled polytope with rational vertices and thus $f_P(n)$ is the Ehrhart quasi-polynomial (see \cite[Chapter 3]{beck2007computing} for an introduction to Ehrhart Theory).

    We now look at the case $I_{\exp} \neq \emptyset$.
    Let $v_e$ be one vertex with $e \in I_{\exp}$ and let $p$ be the gcd of all $p_i$ with $i \in I_1 \cup I_{\exp}$.
    This implies that $\prod_{i=1}^r(A_{S_i}^{\ell_i})_{v_iv_i} = 0$ whenever any of the $\ell_i$ with $i \in [r]$ are not a multiple of $p$ and thus $f_P(n) = 0$ if $n \not\equiv_p r-1$.
    Additionally by Bézout's lemma there is some $N \in \IN$ which is a multiple of $p$ s.t.\ that every number of the form $N + j \cdot p$ with $j \in \IN$ can be written as a non-negative linear combination of the $p_i$ with $i \in I_1 \cup I_{\exp}$.
    Define $N' := N + r - 1 + \sum_{i \in I_{\exp}} N_i$.
    Let $n = N' + j \cdot p$ for any $j \in \IN$.
    We can write
    \begin{align}
        n &= N' + \underbrace{\frac{n - N' - \lfloor \frac{n - N'}{p_e} \rfloor \cdot p_e}{p}}_{\alpha :=} \cdot p + \underbrace{\left\lfloor \frac{n - N'}{p_e} \right\rfloor}_{\beta :=} \cdot p_e\,.
    \end{align}
    since $e \in I_{\exp}$ implies that $p$ divides $p_e$.
    Furthermore both $\alpha$ and $\beta$ are non-negative integers and $\beta$ is maximal with this property.
    In particular we can write $N + \alpha \cdot p$ as a non-negative linear combination with elements $p_i$ with $i \in I_1 \cup I_{\exp}$.
    This gives $N' + \alpha \cdot p = r-1 + \sum_{i \in I_{\exp}} N_i + \sum_{i \in I_1 \cup A_{\exp}} \kappa_i \cdot p_i$ for $\kappa_i \in \IN$.
    We choose $\ell_i = 0$ for $i \in I_0$ and $\ell_i = \kappa_i \cdot p_i$ for $i \in I_1$ and $\ell_i = N_i + \kappa_i \cdot p_i$ for $i \in I_{\exp} \setminus \{e\}$ and $\ell_e = N_e + \kappa_{e} \cdot p_e + \beta \cdot p_e$.
    This describes a subset of the walks of length
    \begin{align*}
        r-1 + \sum_{i \in [r]}\ell_i &= r-1 + \sum_{i \in I_{\exp}} N_i + \sum_{i \in I_{1} \cup I_{\exp}} \kappa_i \cdot p_i + \beta \cdot p_e = n_1 + \alpha \cdot p + \beta \cdot p_e = n\,.
    \end{align*}
    In particular we have $(A_{S_i}^{\ell_i})_{v_iv_i} \geq 1$ for all $i \in [r]$ and $(A_{S_i}^{\ell_i})_{v_iv_i} \in 2^{\Theta(\ell_i)} = 2^{\Theta(n)}$ and thus $f_P(n) \geq 2^{\Omega(n)}$ for any $n \geq n_1$ with $n \equiv_p r-1$.
    Note again that $f_P(n)$ can never grow faster than linearly exponential as it is bounded by the total number of walks through $G$, so $f_P(n)$ is of the form required by the statement of this lemma.

    To finish the proof note that the class of ultimately almost PORC functions is closed under addition and multiplication.
    If we multiply a polynomial constituent with an exponential constituent, then either the polynomial constituent is identical to $0$ or there is some $N \in \IN$, s.t.\ the polynomial is strictly positive for every $n \geq N$.
    We have
    \[
        (A^n)_{uv} = \sum_{P}f_P(n)
    \]
    where the sum is over all paths $P$ starting at $u$ and ending at $v$.
    For any fixed $u, v \in V$ this is a sum of constantly many elements, bounded by $(k \cdot \max(A))^{k-1}$, where $\max(A)$ denotes the maximum entry of $A$.
    We conclude that $(A^n)_{uv}$ is an ultimately almost PORC function.
\end{proof}

\restatebinbasischangemult*
\begin{proof}
    Let $D_{\varphi}$ be the set of dominating terms of $\varphi$.
    We consider a set of dominating terms $D$ to be smaller than another set of dominating terms $D'$, whenever $D \neq D'$ and all terms of in $D$ are dominated by terms in $D'$.
    We prove the statement via induction over $D$.
    In case $D_{\varphi} = \emptyset$ we have $\varphi = 0$ which is trivially a functional closure property.
    So let $D_{\varphi} \neq \emptyset$ and let $a \cdot \binom{x_1}{d_1} \cdots \binom{x_m}{d_m}$ be any term from $D_{\psi}$.
    We use a similar idea to Lemma~\ref{lem:binbasischange} and replace the binomial basis by a shifted binomial basis.
    Consider the terms $a_i \cdot \binom{x_1}{d_1} \cdots \binom{x_{i-1}}{d_{i-1}} \binom{x_i}{d_i-1} \binom{x_{i+1}}{d_{i+1}} \cdots \binom{x_m}{d_m}$ in $\psi$ for all $i \in [m]$ with $d_i \geq 1$.
    These are precisely the (potentially new) dominating terms of $\varphi - a \cdot \binom{x_1}{d_1} \cdots \binom{x_m}{d_m}$.
    In order to ensure the positivity of the $a_i$ we instead consider $\psi := \varphi - a \cdot \binom{x_1 - c}{d_1} \cdots \binom{x_m - c}{d_m}$ for some $c \in \IN$ with $c > -\frac{a_i}{a}$ for all $i \in [m]$ with $d_i \geq 1$.
    We recall that the Chu-Vandermonde identity gives us \mbox{$\binom{x-c}{r} = \sum_{i=0}^r(-1)^{r-i} \binom{r-i+c-1}{r-i} \binom{x}{i}$}.
    Thus the new coefficient $a'_i$ of $\binom{x_1}{d_1} \cdots \binom{x_{i-1}}{d_{i-1}} \binom{x_i}{d_i-1} \binom{x_{i+1}}{d_{i+1}} \cdots \binom{x_m}{d_m}$ in $\psi'$ is $a_i + a \cdot (-c) > 0$.
    The new coefficient of $\binom{x_1}{d_1} \cdots \binom{x_m}{d_m}$ in $\psi$ is zero and thus $D_{\psi}$ is smaller than $D_\varphi$.
    By induction we know that $\psi$ can be written as a positive linear combination of products of shifted binomials and thus $\varphi = \psi + a \cdot \binom{x_1 - c}{d_1} \cdots \binom{x_m - c}{d_m}$ has the same property.
\end{proof}

\restatemultipolynomial*
\begin{proof}
    We prove the statement via induction over the number of variables $m$.
    If $m \leq 1$, then $\varphi$ is a constant polynomial and thus a functional closure property of $\sharpFA$ by Lemma~\ref{lem:polynomial}.
    If $m \geq 2$, and let $f_1, \ldots, f_m \in \sharpFA$.
    Write $\varphi$ as a positive integer linear combination of products of shifted binomials using Lemma~\ref{lem:binbasischangemult} and let $c \in \IN$ be the maximum occurring shift.
    We now distinguish between the $f_i$ assuming values from $\{0, \ldots, c-1\}$ and them assuming values of at least $c$ and all combinations thereof for the different $f_i$.
    Using the fact that all subtractions subtract at most a value of $c$ we can replace the subtractions $x_i - c'$ by $\max(x_i - c', 0)$ without changing the value for $x_i \geq c$ and get the following:
    \begin{align*}
        \psi(f_1(w), \ldots, f_m(w)) = \Bigl(\prod_{i \in [m]} \ONE_{f_i(w) \geq c}\Bigr) \cdot \varphi_{\text{clamped}}(f_1(w), \ldots, f_m(w)) \quad\quad\quad\quad\quad\quad\quad\quad \\
        \quad\quad\quad\quad\quad + \sum_{\emptyset \subsetneq I \subseteq [m]}\sum_{\rho} \Bigl(\prod_{i \in I}\ONE_{f_i(w)=\rho(i)}\Bigr) \cdot \Bigl(\prod_{i \in [m] \setminus I}\ONE_{f_i(w) \geq c}\Bigr) \cdot \psi_{I, \rho}(f_1(w), \ldots, f_m(w))
    \end{align*}
    where $\rho$ sums over all functions $I \to \{0, \ldots, c-1\}$, $\varphi_{\text{clamped}}$ is $\varphi$, but with all subtractions $x_i - c'$ replaced by $\max(x_i - c', 0)$ and $\varphi_{I, \rho}$ is $\varphi$, but all variables $x_i$ with $i \in I$ are replaced by the constants $\rho(i)$.
    The $\varphi_{I, \rho}$ are non-negative integer-valued multivariate polynomials in fewer variables and when replacing any variables by further constants we only obtain polynomials that can also be obtained by replacing variables directly in $\varphi$. 
    Thus by induction the $\varphi_{I, \rho}$ are a functional closure property of $\sharpFA$ and $\varphi(f_1, \ldots, f_m) \in \sharpFA$ by Lemmas~\ref{lem:add},~\ref{lem:mult},~\ref{lem:subtraction},~\ref{lem:const}, and~\ref{lem:binom}.
\end{proof}

\restatemultivarpolyn*
\begin{proof}
    Lemma~\ref{lem:multipolynomial} already proves that all multivariate polynomials of this form are functional closure properties of $\sharpFA$.

    Now let $\varphi: \IN^m \to \IN$ be a multivariate polynomial with rational coefficient and a functional closure property of $\sharpFA$.
    By Theorem~\ref{thm:multivclosure} $\varphi$ can be written as a finite sum of finite products of ultimately PORC functions.
    W.l.o.g.\ let all quasiperiods be the same $p \in \IN$ and all offsets be the same $N \in \IN$ and a multiple of $p$, for example by choosing the lowest common multiple of the quasiperiods and choosing the lowest multiple of $p$ that is bigger than all offsets.
    Then
    \begin{align}
        \varphi(x_1, \ldots, x_m) &= \textstyle\sum_{i=1}^r \prod_{j=1}^m \varphi^{(i,j)}(x_j) \label{eq:multipoly:porcrepr}
    \end{align}
    where $\varphi^{(i,j)}$ is an ultimately PORC function with constituents $\varphi^{(i,j)}_0, \ldots, \varphi^{(i,j)}_{p-1}$.
    Consider inputs on a grid of the form $(N+i_1 \cdot p, N+i_2 \cdot p, \ldots, N+i_m \cdot p)$ for $i_1, \ldots, i_m \in \IN$.
    For all these inputs each of the ultimately PORC functions always use the same constituents, namely the $\varphi^{(i,j)}_0$.
    Thus $\varphi$ agrees with $\varphi'(x_1, \ldots, x_m) := \sum_{i=1}^r \prod_{j=1}^m \varphi^{(i,j)}_0(x_j)$ on all inputs from this grid.
    Furthermore the grid is infinite in all dimensions and thus $\varphi = \varphi'$ by multivariate polynomial interpolation (see \cite{gasca2000polynomial} for an exposition of multivariate polynomial interpolation).
    W.l.o.g.\ assume that $\varphi^{(i,j)}_0$ is not the constant zero function.
    Then it has a positive leading coefficient, otherwise $\varphi^{(i,j)}$ would not be a PORC function.
    Combined with the fact that each of the $\varphi^{(i,j)}_0$ is univariate the product $\prod_{j=1}^m \varphi^{(i,j)}_0(x_j)$ for a fixed $i \in [r]$ has exactly one dominating term when written in the binomial basis.
    Its coefficient is precisely the product of the leading coefficients of the $\varphi^{(i,j)}_0$ and thus positive.
    As a result all dominating terms in the binomial basis of $\varphi'$ (and thus $\varphi$) have positive coefficients, since no cancellations can occur unless the dominating term of a summand dominates the dominating term of another summand, in which case the smaller term would no longer be a dominating term.

    Note that whenever any variable, for example $x_j$, in $\varphi$ is replaced by a constant $c$, then this only changes each function $\varphi^{(i,j)}(x_j)$ and thus $\varphi^{(i,j)}_0(x_j)$ to be the constant $\varphi^{(i,j)}(c)$ which is non-negative.
\end{proof}

\subsection{Missing proofs from section~\ref{sec:promises}}
\restatebinary*
\begin{proof}
    \begin{figure}
        \centering
        \begin{tikzpicture}
            \node[start] (S) {$S$};
            \node[accept, right of=S] (A) {$A$};

            \path
                (S) edge [loop above] node {$0,1$} (S)
                (A) edge [loop above] node {$0,0,1,1$} (A)

                (S) edge [-stealth] node {$1$} (A)
                ;
        \end{tikzpicture}

        \caption{NFA computing the binary the value of its input interpreted as a binary number. Edges with a multiplicity of $2$ are denoted by listing the edge label twice.}
        \label{fig:lem:binary:automaton}
    \end{figure}
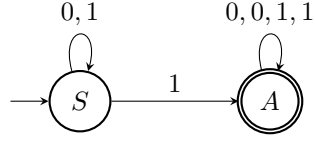
    We define
    \begin{align*}
        a &= \begin{pmatrix}
            1\\
            0\\
        \end{pmatrix}, &
        A_\sigma &= \begin{pmatrix}
            1 & \sigma\\
            0 & 2\\
        \end{pmatrix}, & \text{and} & &
        b &= \begin{pmatrix}
            0\\
            1\\
        \end{pmatrix}\,.
    \end{align*}
    A visualization of the corresponding NFA $M$ via Lemma~\ref{lem:algebraic} can be found in Figure~\ref{fig:lem:binary:automaton}.
    We prove that $M$ computes $f$ by showing via induction that for every $w \in \{0, 1\}^\star$ with $|w| = n$ we have 
    \begin{align*}
        \prod_{i=1}^{n} A_{w_i} = \begin{pmatrix}
            1 & f(w) \\
            0 & 2^n \\
        \end{pmatrix}\,.
    \end{align*}
    For $n=0$ this is obviously correct as $f(\eps) = 0$.
    For the induction step consider the word $w\sigma$ of length $n+1$ for $\sigma \in \{0, 1\}$.
    Then
    \begin{align*}
        \left(\prod_{i=1}^{n} A_{w_i}\right) \cdot A_{\sigma} = \begin{pmatrix}
            1 & f(w) \\
            0 & 2^n \\
        \end{pmatrix} =
        \begin{pmatrix}
            1 & \sigma + 2f(w) \\
            0 & 2^{n+1} \\
        \end{pmatrix} =
        \begin{pmatrix}
            1 & f(w\sigma) \\
            0 & 2^{n+1} \\
        \end{pmatrix}
    \end{align*}
    and we are done.
    Extending the alphabet and ignoring those new symbols can be done by adding a self-loop to each state of $M$.
    This corresponds to an identity matrix and thus doesn't influence the result.
\end{proof}

\restateinterpolatepositive*
\begin{proof}
    We inductively construct integer-valued polynomials $q_0, \ldots, q_N$, that each interpolate one value more.
    We start by setting $q_0(x) := c_0 \cdot \binom{x}{0}$.
    Clearly $q_0(0) = c_0$.
    We proceed inductively by setting $q_{i+1}(x) := q_i(x) + (c_{i+1} - q_i(i+1)) \cdot \binom{x}{i+1}$.
    Note that $\binom{n}{i+1} = 0$ for $n \in \{0, \ldots, i\}$ and thus $q_{i+1}$ and $q_i$ agree on all such $n$.
    The final $q$ is then constructed as $q(x) := q_N(x) + \alpha \cdot \binom{x}{N+1}$, for some big enough $\alpha \in \IN$, s.t.\ $q(n') \geq 0$ for all $n' \in \IN$.
    Such an $\alpha$ always exists since $\binom{x}{N+1}$ has higher degree than $q_N$.
\end{proof}

\restatepolyclusterseq*
\begin{proof}
    If such a $\varphi'$ exists, we directly see that $\varphi$ is a functional promise closure property of $\sharpFA$ with respect to $S$.
    Indeed for any $f_1, \ldots, f_m \in \sharpFA$ we construct $g = \varphi' \circ (f_1, \ldots, f_m) \in \sharpFA$ and see $g(w) = \varphi(f_1(w), \ldots, f_m(w))$ for every $x \in \Sigma^\star$ with $(f_1(w), \ldots, f_m(w)) \in S$.

    Now on the other hand let $\varphi$ be any functional promise closure property of $\sharpFA$ with regard to $S$.
    Again define the alphabet $\Sigma = \{\sigma_1, \ldots, \sigma_m\}$ and the functions $f_i: \Sigma^\star \to \IN$ where $f_i(w) := \#_i(w$ is defined as the number of occurences $\#_i(w)$ of the symbol $\sigma_i$ in $w$.
    Applying the closure property to $f_1, \ldots, f_m$ gives that there is some $g \in \sharpFA$ with $g(w) = \varphi(f_1(w), \ldots, f_m(w))$ for all $w \in \Sigma^\star$ with $(f_1(w), \ldots, f_m(w)) \in S$.
    Let $M = (Q, \Sigma, \wt, \inp, \outp)$ be an NFA computing $g$.
    This induces transition matrices $A_\sigma \in \IN^{|Q| \times |Q|}$ for each symbol $\sigma \in \Sigma$ and vectors $a, b \in \IN^{|Q|}$, s.t.\ $a^T \prod_{j=1}^{|w|}A_{w_j}b = g(x)$ for all $w \in \Sigma^\star$.
    Restricting to words of the form $w = \sigma_1^{n_1} \sigma_2^{n_2} \cdots \sigma_m^{n_m}$ for $n_1, \ldots, n_m \in S$ gives $a^T \left(\prod_{i=1}^{m}A_{\sigma_i}^{n_i}\right) b = \varphi(n_1, \ldots, n_m)$.
    Using Lemma~\ref{lem:powerstructure} on each of the $A_i^{n_i}$ we see that every entry of $A_i^{n_i}$ is an ultimately almost PORC function in $n_i$.
    Consequently every entry of $\prod_{i=1}^{m}A_{\sigma_i}^{n_i}$ is a finite sum of products of different ultimately almost PORC functions and the same holds for $\chi(n_1, \ldots, n_m) := a^T \left(\prod_{i=1}^{m}A_{\sigma_i}^{n_i}\right) b$.
    Note that $\chi_{|S} = \varphi_{|S}$.

    Remains to show that we can rewrite $\chi$ as some $\psi$ with $\psi_{|S} = \chi_{|S}$, s.t.\ none of the constituents are exponential.
    Again we do this by looking at each summand individually and replacing one exponential constituent at a time.
    For this let $\chi^{(1)}(n_1) \cdots \chi^{(m)}(n_m)$ be one of the summands of $\chi$ where $\chi^{(1)}, \ldots, \chi^{(m)}$ are all ultimately almost PORC functions, with periods $p_1, \ldots, p_m$, offsets $N_1, \ldots, N_m$ and constituents $\chi^{(i)}_0, \ldots, \chi^{(i)}_{p_i-1}$ for each $i \in [m]$.
    W.l.o.g.\ choose the offsets big enough, such that every constituent $\chi^{(i)}_j$ is either constant zero for $n_i \geq N_i$ or never zero for any $n_i \geq N_i$.
    Let $\chi^{(i)}_j$ be one of the exponential constituents and let $\gamma \in \IR^+$ and $N \in \IN$ be, s.t. $\chi^{(i)}_j(n_i) \geq 2^{\gamma n_i}$ for $n_i \geq N$.
    Let $B_{\max, i}$ be the maximum the $i$-th coordinate assumes in any preimage of a singleton set under $\tau$ where the $i$-th coordinate is bounded.
    We replace $\chi^{(i)}_j$ with a polynomial that interpolates $\chi^{(i)}$ for all values of at most $B_{\max, i}$.
    Lemma~\ref{lem:interpolate_positive} ensures this is possible with an integer-valued polynomial that is non-negative for all of $\IN$.
    Call the resulting functions after replacement $\psi^{(i)}$ and $\psi^{(i)}_j$.
    We claim that this does not affect the value of the product for any $(n_1, \ldots, n_m) \in S$.
    For the sake of contradiction the value of the product did change for some $(c_1, \ldots, c_m) \in S$.
    We thus have
    \[
        \chi^{(1)}(c_1) \cdots \chi^{(m)}(c_m) \neq \chi^{(1)}(c_1) \cdots \chi^{(i-1)}(c_{i-1})\cdot \psi^{(i)}(c_i) \cdot \chi^{(i+1)}(c_{i+1}) \cdots \chi^{(m)}(c_m)
    \]
    which implies $\chi^{(1)}(c_1) \cdots \chi^{(i-1)}(c_{i-1}) \cdot \chi^{(i+1)}(c_{i+1}) \cdots \chi^{(m)}(c_m) \neq 0$ and $c_i \geq N_i$ as no function except $\chi^{(i)}$ for inputs at least as big as $N_i$ has been changed.
    By the same reason we know $\chi^{(i)}_{j}(c_i) \neq \psi^{(i)}_{j}(c_i)$ and $c_i \equiv_{p_i} j$.

    Consider the preimage $T = \tau^{-1}(\smod{c_1}{p_1}{N_1}, \ldots, \smod{c_m}{p_m}{N_m})$.
    Clearly $(c_1, \ldots, c_m) \in T$.
    The $i$-th coordinate of $T$ has to be unbounded, because $c_i \leq B_{\max, i}$ would imply $\psi^{(i)}_{j}(c_i) = \chi^{(i)}(c_i) = \chi^{(i)}_{j}(c_i)$ by the construction of $\psi^{(i)}_j$.
    Furthermore the behaviour of each of the $\chi^{(k)}$ is only dictated by at most one of the constituents each.
    By our choice of the offsets we have $\chi^{(k)}(n_k) \neq 0$ for all $k \neq i$ and all $(n_1, \ldots, n_m) \in T$.
    Since $S$ admits polynomial cluster sequences there is a polynomial cluster sequence $T' \subseteq T$ with regards to dimension $i$ with polynomial $q: \IN \to \IN$.

    We construct $m$ functions $f'_k \in \sharpFA$ over the alphabet $\Sigma = \{\sigma_{1,0}, \sigma_{1,1}, \ldots, \sigma_{m,0}, \sigma_{m,1}\}$ with $2m$ symbols, by using Lemma~\ref{lem:binary} on the symbols $\sigma_{k,0}$ and $\sigma_{k,1}$ each while ignoring any of the other symbols.
    This allows us to encode any $(n_1, \ldots, n_m) \in T'$ as independent binary representations without leading zeros for each $n_k$ into a string of total length $\sum_{k=1}^m \lfloor\log_2(n_k) + 1 \rfloor \leq m \cdot \lfloor \log_2(q(n_i)) + 1\rfloor$ which is in $O(\log(n_i))$ for fixed $m$ and $q$.

    Now let $w \in \Sigma^\star$ be any such string corresponding to a value $(n_1, \ldots, n_m) \in T'$.
    Combining all of this we again reach a contradiction to the fact that NFAs can only compute at most linearly exponential functions via
    \begin{align*}
        \varphi(f'_1(w), \ldots, f'_m(w)) &\geq \chi^{(1)}(c_1) \cdots \chi^{(i-1)}(c_{i-1}) \cdot \chi^{(i)}(f'_i(w)) \cdot \chi^{(i+1)}(c_{i+1}) \cdots \chi^{(m)}(c_m)\\
                                  &\geq \chi^{(i)}(f'_i(w))\\
                                  &= \chi^{(i)}_j(f'_i(w))\\
                                  &= \chi^{(i)}_j(n_i)\\
                                  &\geq 2^{\gamma n_i}\,.\\
    \end{align*}
    Note that $|w| \in O(\log(n_i))$ implies $2^{\gamma n_i}$ is doubly exponential in the input length.
\end{proof}

\restatemonotonegraph*
\begin{proof}
    Let $N_1, \ldots, N_m \in \IN$ and $p_1, \ldots, p_m \in \IN$ be fixed and consider the preimage $T$ of any singleton $(c_1, \ldots, c_m) \in \{0, \ldots, N_1 + p_1 - 1\} \times \ldots \times \{0, \ldots, N_m + p_m - 1\}$ under the corresponding shifted grid projection $\tau$ of $S$.
    Let $(n_1, \ldots, n_m) \in T$ and let $i \in [m]$ be arbitrary.
    If the $i$-th coordinate in $T$ is bounded then we are done.
    If the $i$-th coordinate in $T$ is unbounded we distinguish two cases.
    
    In case $1 \leq i \leq j$, we choose the $T'$ to be the part of $S$ spanned by choosing the values of the free variables as \mbox{$\{(n_1, \ldots, n_{i-1}, n_i+\kappa \cdot pd, n_{i+1}, \ldots, n_j) \mid \kappa \in \IN \}$} where $p$ is the lcm of $p_1, \ldots, p_m$ and $d$ is the common denominator of all of the $\mu_1, \ldots, \mu_k$.
    Clearly the $i$-th coordinate of $T'$ is unbounded and $T'$ is infinite.
    Further since all other free variables are constant, each dependent variable is a monotone univariate polynomial in the $i$-th coordinate and thus the sum of all coordinates is bounded by some polynomial in the $i$-th coordinate.
    Remains to show $T' \subseteq T$.
    For this let $(n'_1, \ldots, n'_m) \in T'$.
    Clearly $\smod{n'_\ell}{p_\ell}{N_\ell} = \smod{n_\ell}{p_\ell}{N_\ell}$ for $\ell \in [j] \setminus \{i\}$, since $n'_\ell = n_\ell$.
    Additionally $\smod{n'_i}{p_i}{N_i} = \smod{n_i}{p_i}{N_i}$, since $n'_i = n_i + \kappa \cdot pd$ and $n_i \geq N_i$.
    Furthermore $\smod{n'_\ell}{p_\ell}{N_\ell} = \smod{n_\ell}{p_\ell}{N_\ell}$ for $\ell \in [m] \setminus [j]$, because either $n'_\ell = n_\ell$ or the $\ell$-th coordinate is unbounded and thus $n'_\ell \geq n_\ell \geq N_\ell$ and $n'_\ell \equiv_{p_\ell} n_\ell$ since $n'_\ell$ is dependent on $n'_i$ as a univariate polynomials and for any univariate polynomial $q(x)$ we have $q(x) \equiv_{p_\ell} q(x + \delta)$ if $\delta$ is a multiple of $p_\ell \cdot h$.
    Thus $(n'_1, \ldots, n'_m) \in T$.

    Remains to consider the case where $i > j$.
    Since the $i$-th coordinate in $T$ is unbounded, $n_i$ is a function depending on at least one of the free variables $n_{i'}$.
    We repeat the previous argument by continuing with $i'$ instead of $i$ which finishes the proof, since $n_i \geq \alpha \cdot n_{i'}$ for some $\alpha \in \IR^+$.
\end{proof}

\restatemongraphclosures*
\begin{proof}
    If there exists a $\psi \in I$ such that $\varphi + \psi$ is a multivariate functional closure property of $\sharpFA$, then $(\varphi + \psi)_{|S} = \varphi_{|S}$ and thus $\varphi$ is a functional promise closure property of $\sharpFA$ with regard to $S$ by definition.
    On the other hand if $\varphi$ is a functional promise closure property of $\sharpFA$ with regard to $S$, then there is a functional closure property $\varphi'$ of $\sharpFA$ with $\varphi_{|S} = \varphi'_{|S}$.
    Note that a priori $\varphi'$ might not be a polynomial, but just a finite sum of finite products of univariate ultimately PORC functions.

    Wlog all thresholds $N$ and quasiperiods $p$ of $\varphi'$ coincide. We write $p^{\times m}=(p,p,\ldots,p)$ and $N^{\times m}=(N,N,\ldots,N)$.
    For each $\vv i \in \{0, \ldots, p+N-1\}^m$ we have a polynomial $\varphi'_{\vv i}$, and $\varphi'(\vv n) = \varphi'_{\smodtimesm{\vv n}{p}{N}}(\vv n)$ for all $\vv n \in \IN^m$.
    Define the set $S'=\vv\mu^{-1}(\IN^m)\subseteq\IN^j$ as all $\vv s \in \IN^j$, s.t.\ $\mu_a(\vv s) \in \IN$ for all $a \in [k]$.
    Let $\widetilde\varphi'_{\vv i} : \IQ^j\to\IQ$ denote the pullback via $\mu$, i.e., $\widetilde\varphi'_{\vv i}(s_1,\ldots,s_j) = \varphi'_{\vv i}(s_1,\ldots,s_j,\mu_1(\vv s),\ldots,\mu_k(\vv s))$.
    Analogously define the polynomial $\widetilde\varphi: S'\to\IQ$.
    Since $\varphi'$ is only defined on $\IN^m$, the definition of its pullback requires a bit more care, but since the pullback of $\varphi'$ is defined for all of $S'$, we can then analogously define $\widetilde\varphi': S'\to\IQ$.
    Observe that $\widetilde\varphi$ and $\widetilde\varphi'$ coincide on $S'$.
    Let $d$ be the common denominator of $\mu_1, \ldots, \mu_k$.
    Note that for any fixed $a\in[k]$ and $\vv o \in S'$ we have that $\mu_a(\vv o + dp \vv b) \rem p$ is constant for all $\vv b \in \IN^j$.
    In particular we also have $\vv o + dp \vv b \in S'$.
    Due to the monotonicity of each of the $\mu_a$, for big enough $\vv b$, $\mu_a(\vv o + dp \vv b)$ will either always be at least $N$ or $\mu_a$ is a constant function.
    Thus there is now some $\vv o \in S'$ with $\vv o \geq N^{\times j}$, such that even $\smod{\mu_a(\vv o + dp \vv b)}{p}{N}$ is constant for all $\vv b \in \IN^j$.
    Let $\vv{\theta} = \smod{(\vv o,\mu_1(\vv o),\ldots,\mu_k(\vv o))}{p^{\times m}}{N}$.
    Hence, for all $\vv b \in \IN^{j}$ we have
    $\widetilde \varphi'(\vv o + dp \vv b) = \widetilde\varphi'_{\vv{\theta}}(\vv o + dp \vv b)$.
    Since $\widetilde\varphi$ and $\widetilde\varphi'$ coincide on $S'$,
    it follows that
    $\widetilde \varphi$ and $\widetilde\varphi'_{\vv{\theta}}$
    coincide on an infinitely large full-dimensional subgrid (namely on the grid formed by $\vv o + dp \vv b$ for $b \in \IN^j$).
    Hence by multivariate polynomial interpolation, $\widetilde \varphi = \widetilde\varphi'_{\vv{\theta}}$.
    Hence, $\varphi_{|S} = \varphi'_{\vv \theta|S}$.
    Their difference
    $\varphi'_{\vv \theta} - \varphi$
    is the desired polynomial $\psi\in I$.
    It remains to show that $\varphi'_{\vv \theta}$ is a multivariate closure property of $\sharpFA$.
    Repeating the same argument as in Lemma~\ref{lem:multivarpolyn} to proves that the dominating terms of $\varphi'_{\vv \theta}$ have positive coefficients, even after replacing some variables by constants.
    Using Lemma~\ref{lem:multivarpolyn} again, shows that $\varphi'_{\vv \theta}$ is a multivariate closure property of $\sharpFA$.
\end{proof}

\end{document}